\documentclass{article}
\usepackage[english]{babel}
\usepackage{amsfonts,amsmath,amssymb}
\usepackage{graphicx}
\usepackage{bm}
\usepackage{tabularray}
\usepackage[letterpaper,top=2cm,bottom=2cm,left=3cm,right=3cm,marginparwidth=1.75cm]{geometry}
\usepackage{enumerate}  
\usepackage{tikz}
\usepackage{bbold}
\usepackage{amsmath}
\usepackage{graphicx}
\usepackage[colorlinks=true, allcolors=blue]{hyperref}
\usepackage{tikz}
\usepackage{multirow}
\usepackage{array}
\usepackage{fancyvrb}

\usepackage{amsthm}
\usepackage{natbib}
\setcitestyle{authoryear}
\usepackage[font = small]{caption}

\newtheorem{example}{Example}
\newtheorem{theorem}{Theorem}
\newtheorem{extended-theorem}{Extended Theorem}
\newtheorem{lemma}{Lemma}

\newtheorem{remark}{Remark}
\usepackage{mathtools}  
\mathtoolsset{showonlyrefs}

\RecustomVerbatimEnvironment{Verbatim}{Verbatim}{xleftmargin=7mm}
\DeclareMathOperator*{\argmin}{argmin}

\usepackage{scalerel,stackengine}
\stackMath

\usepackage{xcolor}
\usepackage{booktabs}
\usepackage{arydshln}

\title{Out-of-distribution generalization under random, dense
distributional shifts}
\author{Yujin Jeong and Dominik Rothenh\"ausler}

\begin{document}

\maketitle

\begin{abstract}
    Many existing approaches for estimating parameters in settings with distributional shifts operate under an invariance assumption. For example, under covariate shift, it is assumed that $p(y|x)$ remains invariant. We refer to such distribution shifts as sparse, since they may be substantial but affect only a part of the data generating system. In contrast, in various real-world settings, shifts might be dense. More specifically, these dense distributional shifts may arise through numerous small and random changes in the population and environment. First, we discuss empirical evidence for such random dense distributional shifts. Then, we develop tools to infer parameters and make predictions for partially observed, shifted distributions. Finally, we apply the framework to several real-world datasets and discuss diagnostics to evaluate the fit of the distributional uncertainty model.
\end{abstract}

\section{Introduction}

Distribution shift is a persistent challenge in statistics and machine learning. For example, we might be interested in understanding the determinants of positive outcomes in early childhood education. However, the relationships among learning outcomes and other variables can shift over time and across locations. This makes it challenging to gain actionable knowledge that can be used in different situations. As a result, there has been a surge of interest in robust and generalizable machine learning. Existing approaches can be classified as invariance-based approaches and adversarial approaches.  
 
Invariance-based methods assume that the shifts only affect a sub-space of the data distribution \citep{pan2010transfer,baktashmotlagh2013unsupervised}. For example, under covariate shift \citep{quinonero2009dataset}, the distribution of the covariates changes while the conditional distribution of the outcomes given the covariates stays invariant. In causal representation learning \citep{scholkopf2021toward}, the goal is to disentangle the distribution into independent causal factors; based on the assumption that spurious associations can be separated from invariant associations. Both approaches assume the shift affects only a sub-space of the data. We call such shifts \emph{sparse} distribution shifts.

Adversarial methods in robust machine learning consider \emph{worst-case} perturbations. The idea is that some adversary can perturb the data either at the training \citep{huber81} or deployment stage \citep{szegedy-2014,duchi2021learning}; and the goal is to learn a model that is robust to such deviations. While developing robust procedures is important, worst-case perturbations might be too conservative in many real-world applications. 

Motivated by an empirical observation, in this paper we consider distributional shifts due to \emph{random} and \emph{dense} perturbations. As an example, consider replication studies. In replication studies, different research teams try to replicate results by following the same research protocol. Therefore, any distribution shift between the data may arise from the superposition of many unintended, subtle errors that affect every part of the distribution, which might be modeled as random. We call such shifts \textit{dense} distributional shifts. The random, dense distributional shift model has recently been used to quantify distributional uncertainty in parameter estimation problems \citep{jeong2023calibrated, rothenhausler2023distributionally} and has shown success in modelling real-world temporal shifts in refugee placement \citep{bansak2023random}.  We will develop tools to measure the similarity between randomly perturbed distributions, infer parameters and predict outcomes for partially observed and shifted distributions.

\subsection{Summary of contributions}

In this subsection, we preview and summarize our contributions.

First, we discuss an intriguing phenomenon on a real-world dataset, in which means of arbitrary functions seem to be correlated across multiple distributions. To the best of our knowledge, existing distribution shift models do not imply this phenomenon.

To address this gap, we introduce a random distribution shift model for multiple datasets, where likelihood ratios exhibit random fluctuations. This builds on the perturbation model introduced in \cite{jeong2023calibrated}, which considers a single perturbed distribution. We generalize the model to handle multiple perturbed distributions that are possibly correlated. We then demonstrate the extended random perturbation model can explain the empirical distribution shift phenomenon.

Next, we consider the following domain adaptation problem. Our goal is to predict outcomes \( Y \) based on covariates \( X \). Specifically, we seek to estimate  
$\theta^0 = \arg \min_\theta \mathbb{E}^0[\mathcal{L}(\theta,X,Y)]$ 
for some loss function \( \mathcal{L}(\theta,X,Y) \), where \( \mathbb{E}^0 \) denotes the expectation under the unknown target distribution. While we have access to complete datasets \( (X,Y) \) from multiple training distributions, the target dataset contains only covariates \( X \) without corresponding outcomes \( Y \).  
We adopt a weighted Empirical Risk Minimization (ERM) approach, where the parameters $\theta^0$ are estimated as
\begin{equation*}
    \hat{\theta}^0 = \sum_{k=1}^K \beta_k\hat{\mathbb{E}}^{k}[\mathcal{L}(\theta, X, Y)],
\end{equation*}
where $\hat{\mathbb{E}}^k$ denotes the empirical mean computed on the $k$-th training dataset. 

Under the random distribution shift model, we derive optimal weights for weighted ERM and expressions for the corresponding out-of-distribution error. Based on this expression, we derive prediction intervals for parameters of the target distribution.

\begin{figure}[t!]
    \centering
    \includegraphics[scale = 0.5]{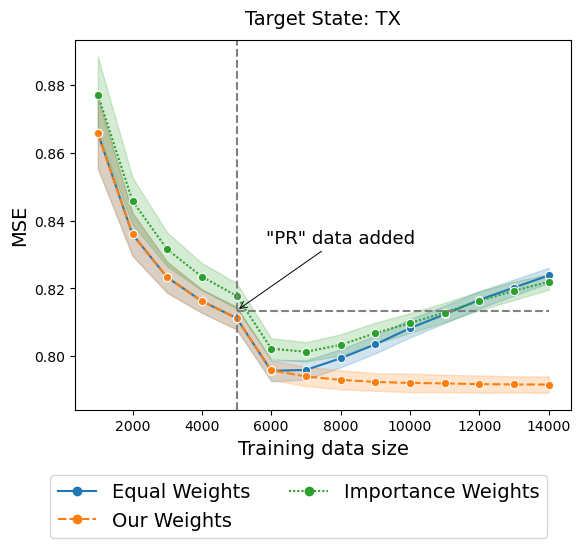}
    \includegraphics[scale = 0.5]{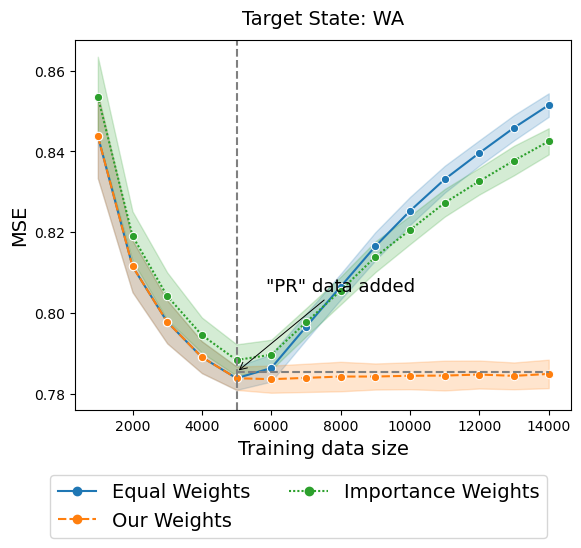}
    \caption{Test MSE results of the XGBoost when the training data is initially sourced from CA and then sourced from PR. The dashed vertical line indicates the point when PR data started to be added. The blue line is when samples are equally weighted, and the green line is when samples are weighted based on importance weights. The orange line is when samples receive distribution-specific weights using our proposed method.}
    \label{fig:ACS}
\end{figure}

The effectiveness of our approach is illustrated in Figure~\ref{fig:ACS}, which shows the out-of-distribution error for income prediction task (see Section~\ref{sec:acs} for details) when the target state is either Texas (TX) or Washington (WA) state, and training datasets are from California (CA) and Puerto Rico (PR). We simulate a scenario where CA data is initially available, followed by additional data from PR. 
We optimize the weighted empirical risk under three different weighting schemes: (1) equal weighting of all samples, (2) sample-specific importance weighting under the assumption of covariate shifts, and (3) distribution-specific weighting based on our approach. Compared to equal weighting and importance weighting methods, our approach demonstrates remarkable robustness, maintaining consistently low out-of-distribution error even after including the PR dataset by optimally weighting two data sources.

The remainder of paper is outlined as follows. In Section~\ref{sec:distributional-uncertainty}, we present empirical evidence for random and dense distributional shifts and introduce the random distributional perturbation model.
In Section~\ref{sec:distribution-generalization}, we establish the foundations for domain adaptation under random and dense distributional shifts. In particular, we derive out-of-distribution risk for weighted empirical risk minimization and conduct statistical inference for parameters of the target distribution. In Section~\ref{sec:acs} and Section~\ref{sec:gtex}, we apply our framework to several real-world datasets and demonstrate the robustness of our method. We conclude in Section~\ref{sec:discussion}.

\subsection{Related work}

This work addresses domain adaptation problems under a specific class of random distributional shifts. Domain adaptation methods generally fall into two main categories \citep{pan2010transfer}. The first type involves re-weighing, where training samples are assigned weights that take the distributional change into account \citep{long2014transfer, li2016prediction}. For example, under covariate shift, training samples can be re-weighted via importance sampling \citep{gretton2009covariate, shimodaira2000improving, sugiyama2008direct}. Another type aims to learn invariant representations of features \citep{argyriou2007, baktashmotlagh2013unsupervised}.
Under dense distributional shifts, as we will see in Section~\ref{sec:distributional-uncertainty}, instance-specific weights can be unstable or even be undefined due to overlap violations. Furthermore, under dense distributional shifts there might be no invariant representation. Thus, from a classical domain adaptation perspective, dense distribution shifts are a challenging setting. 

The proposed method approximates the target distribution as a weighted combination of source distributions. Therefore, the proposed method shares methodological similarities with synthetic controls, which search for a linear combination of untreated units to represent treated units \citep{abadie2003basque}. 
Synthetic controls are applied to panel data and often assume a linear factor model for a continuous outcome of interest \citep{doudchenko2016balancing, abadie2007synthetic, bai2009panel, xu2017generalized, athey2021matrix}. Our framework can be seen as a justification for synthetic control methods under random distributional shifts. In addition, our procedure can be applied  to model distribution shifts with or without time structure for any type of data (discrete, continuous, ordinal).  Furthermore, our framework allows for straightforward generalizations to empirical risk minimization, which includes many modern machine learning tools.

The approach of weighting a combination of source distributions to approximate a target distribution also appears in the distributionally robust optimization (DRO) framework \citep{xiong2023distributionally, zhan2024transfer}.  While there is some similarity in the prediction procedures, the random distribution shift model enables diagnostics and inference that are markedly different from a worst-case approach. One can think about the random shift model as a potential justification for a very particular type of distribution weighting, with additional inferential consequences.

Our procedure is also loosely related to meta-analysis and hierarchical modeling. Meta-analysis is a statistical procedure for combining data from multiple studies \citep{ higgins2009REmetaanalysis}. Traditional meta-analysis relies on the assumption that the effect size estimates from different studies are independent. However, in real-world applications, this independence assumption often does not hold true. 
There are different strategies to handle dependencies. These include averaging the dependent effect sizes within studies into a single effect \citep{wood2008averagingeffects}, incorporating dependencies into models through robust variance estimation \citep{hedges2010robust}, or through multilevel meta-analysis where they assume a cluster structure and model dependence within and between clusters \citep{cheung2014modeling}. In meta-regression, the effect sizes are expressed as linear functions of study characteristics, and they are independent conditional on these study characteristics \citep{borenstein2009book, hartung2008book}. Bayesian hierarchical models with prior on model parameters can also be used to model dependency  \citep{higgins2009REmetaanalysis, lunn2013fullybayesian}. 
 Compared to these procedures, our approach allows for correlations between studies and can be applied even when there is no hierarchical structure and no summary study characteristics are available. Furthermore, our approach is frequentist, so it does not require choosing a prior distribution.

In this paper, we model multiple correlated distributions using random distributional perturbation model introduced in \cite{jeong2023calibrated}, where dense distributional shifts emerge as the superposition of numerous small random changes.  In \cite{jeong2023calibrated}, a single perturbed distribution is considered. In our paper, we extend this model to address multiple perturbed distributions that are possibly correlated.

\section{Distributional Uncertainty}\label{sec:distributional-uncertainty}

Due to a superposition of many random, small changes, a data scientist might not draw samples from the target distribution $\mathbb{P}^0$ but from several randomly perturbed distributions $\mathbb{P}^{1}, \dots, \mathbb{P}^{K}$. More specifically, we assume that a data scientist has $K$ datasets, with the $k$-th dataset being an i.i.d.\ sample from a perturbed distribution $\mathbb{P}^k$. Intuitively speaking, some of these distributions might be similar to each other, while others might be very different. 
For instance, one replication study might closely resemble another due to collaboration among investigators. Additionally, a dataset obtained from California might be more similar to one from Washington than to one from Wyoming. 

We illustrate the scenario where multiple datasets exhibit distributional correlations 
using the GTEx\footnote{The Genotype-Tissue Expression (GTEx) Project was supported by the Common Fund of the Office of the Director of the National Institutes of Health, and by NCI, NHGRI, NHLBI, NIDA, NIMH, and NINDS. The dataset used for the analyses described in this manuscript is version 6 and can be downloaded in the GTEx Portal: \url{www.gtexportal.org}.} data. The GTEx V6 data offers pre-processed RNA-seq gene-expression levels collected from 450 donors across 44 tissues. 
Treating each tissue as an individual study, we have 44 datasets. Each dataset from tissue $k$ is considered as an independent sample from a perturbed distribution $\mathbb{P}^k$, where $\mathbb{P}^{\bullet}$ represents a joint distribution of all gene-expression levels. Some of these datasets are likely to be correlated due to factors such as an overlap of donors and the proximity between tissues. Empirical evidence of this correlation is presented in the left side of Figure \ref{fig:gtex_pairs_plot}, where each row and column corresponds to a tissue, and each dot represents a standardized covariance between gene expression levels of a randomly selected gene-pair, among the observations in the tissue. We randomly sample 1000 gene-pairs. We can see that some tissues exhibit a clear linear relationship. In the following, we will introduce a distribution shift model that implies this linear relationship.

\begin{figure}[t!]
    \centering
    \begin{minipage}{.55\textwidth}
      \centering
      \includegraphics[scale = 0.65]
    {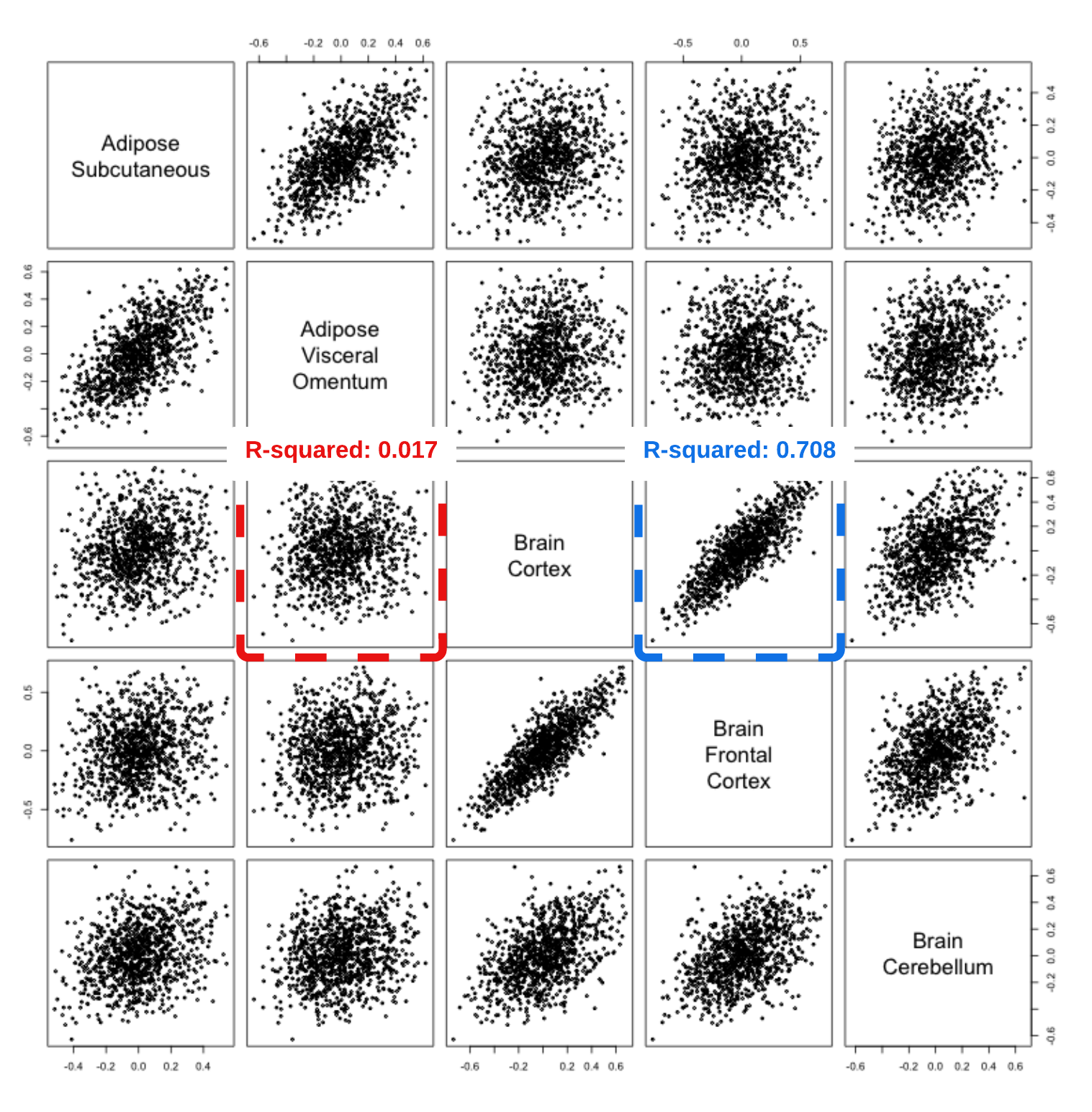}
    \end{minipage}
    \begin{minipage}{.4\textwidth}
      \centering
     \includegraphics[scale = 0.45]{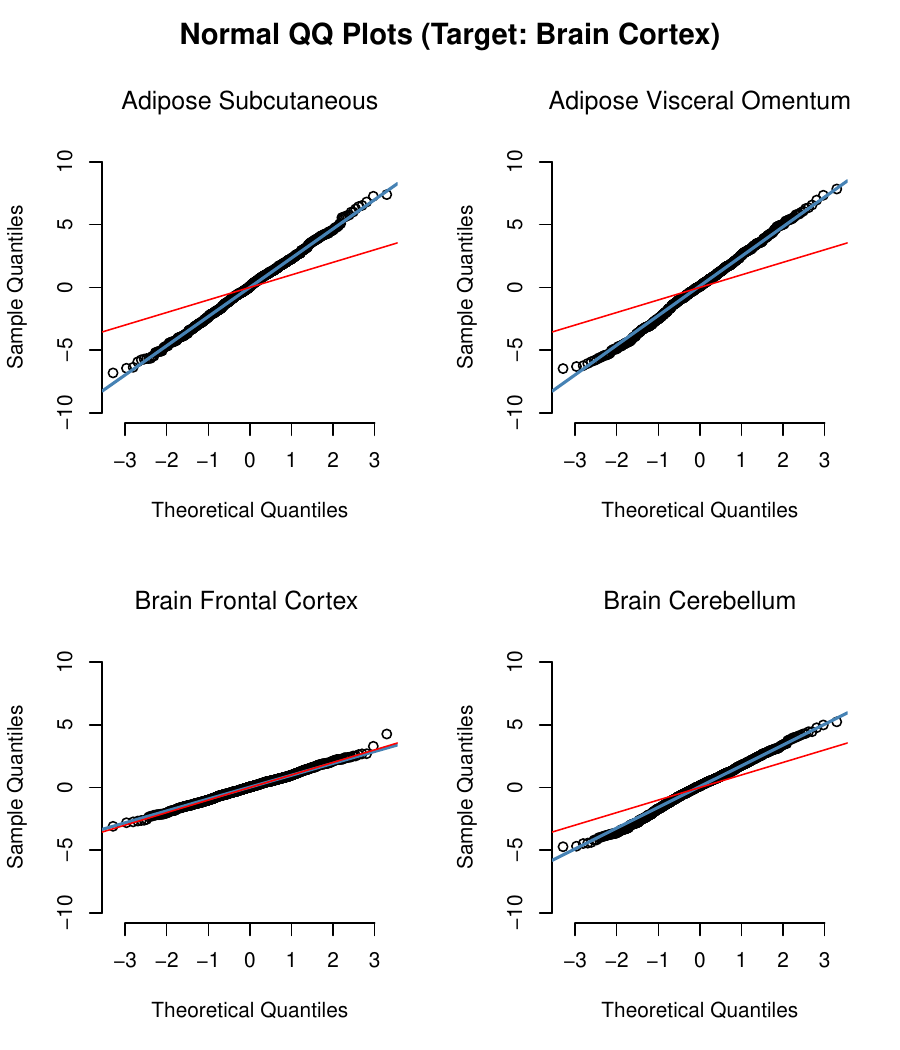}
    \end{minipage}
    \caption{\textbf{(Left)} Linear relationships between tissues: Each row and column corresponds to one tissue, with each dot representing the standardized covariance of the gene expression levels of a randomly selected gene pair. Note that certain tissues exhibit a linear relationship. For instance, a notable linear relationship is observed within the blue squared box, where a simple linear regression yielded an R-squared value of 0.708. In contrast, in the red squared box, a linear relationship is not evident, as indicated by the simple linear regression result with an R-squared value of 0.017.
    \textbf{(Right)} QQ plots of \eqref{eq:ttest} for various test functions for different perturbed tissues when the target tissue is brain cortex. The red line represents the expected QQ line if the perturbed data were all drawn i.i.d.\ from the target distribution (brain cortex).}
    \label{fig:gtex_pairs_plot}
\end{figure}

To model the relatedness of multiple  datasets, we will assume that each $\mathbb{P}^k$ is a random perturbation of a common target distribution $\mathbb{P}^0$. Then, 
we adopt the random distributional perturbation model in \cite{jeong2023calibrated}, where dense distributional shifts emerge as the superposition of numerous small random changes, to construct perturbed distributions $\mathbb{P}^{1}, \dots, \mathbb{P}^{K}$. In this random perturbation model, probabilities of events from a random probability measure $\mathbb{P}^k$ are slightly up-weighted or down-weighted compared to the target distribution $\mathbb{P}^0$. 

Without loss of generality, we construct distributional perturbations for uniform distributions on $[0,1]$. A result from probability theory shows that any random variable $D$ on a finite or countably infinite dimensional probability space can be written as a measurable function $D \stackrel{d}{=} h(U)$, where $U$ is a uniform random variable on $[0,1]$.\footnote{For any Borel-measurable random variable $D$ on a Polish (separable and completely metrizable) space $\mathcal{D}$, there exists a Borel-measurable function $h$ such that $D \stackrel{d}{=} h(U)$ where $U$ follows the uniform distribution on $[0,1]$ \citep{dudley2018real}.} With the transformation $h$, the construction of distributional perturbations for uniform variables can be generalized to the general cases by using
\begin{equation*}
        \mathbb{P}^{k}(D \in \bullet) = \mathbb{P}^{k}(h(U) \in \bullet).
\end{equation*}
Let us now construct the perturbed distribution $\mathbb{P}^{k}$ for a random variable $U$ that follows a uniform distribution under $\mathbb{P}^0$. Take $m$ bins $I_j = [(j-1)/m, j/m]$ for $j = 1, \dots, m$. Let $W^k_1, \dots, W^k_m$ be i.i.d.\ positive random variables with finite variance. 
We define the randomly perturbed distribution $\mathbb{P}^{k}$ by setting
$$
    \mathbb{P}^{k}(U \in \bullet) = \sum_j \mathbb{P}^0(U \in I_j \cap \bullet) \cdot \frac{W^k_j}{\sum_{j=1}^{m} W^k_j/m}.
$$ 
Then the datasets are generated as follows:  

1. Generate perturbed distributions: Start with the target distribution \(\mathbb{P}^0\) and generate 
\(K\) perturbed distributions \(\mathbb{P}^1, \dots, \mathbb{P}^K\) using random weights \(W = \{\{W_j^k\}_{j=1}^{m}\}_{k=1}^K\).

2. Draw samples: For each \(k = 1, \dots, K\), draw \(n_k\) i.i.d. samples, $\mathcal{D}_k = \{D_{ki}\}_{i=1}^{n_k}$, from the perturbed distribution \(\mathbb{P}^k\) ($D_{ki} \stackrel{i.i.d}{\sim}\mathbb{P}^k$), conditioned on the weights  \(W = \{\{W_j^k\}_{j=1}^{m}\}_{k=1}^K\).

Further, let $m = m(n)$ such that $\frac{m(n)}{\min(n_1, \dots, n_K)}$ converges to 0. This is the regime where the distributional uncertainty is of higher order than the sampling uncertainty.

\paragraph{Notation.} Let $\mathbb{P}^0$ be the target distribution on $\mathcal{D}$, and $\mathbb{P}^k$ be the perturbed distribution on $\mathcal{D}$ conditioned on $(W_j^k)_{j = 1}^m$. We draw an i.i.d. sample $\mathcal{D}_{k} = \{D_{ki}\}_{i=1}^{n_k}$ from $\mathbb{P}^k$, conditioned on $(W_j^k)_{j = 1}^m$. Let $P$ be the marginal distribution of $\{\{D_{ki}\}_{i=1}^{n_k},\{W_j^k\}_{j=1}^m\}_{k=1}^K$.
We denote $E$ as the expectation under $P$, $\mathbb{E}^0$ as the expectation under $\mathbb{P}^0$, and $\mathbb{E}^{k}$ as the expectation under $\mathbb{P}^{k}$. Specifically,  $\mathbb{E}^0[\phi(D)]$ is the mean of $\phi(D)$ for a random variable $D \in \mathcal{D}, D \sim \mathbb{P}^0$ and $\mathbb{E}^k[\phi(D)]$ is the mean of $\phi(D)$ for a random variable $D \in \mathcal{D}, D \sim \mathbb{P}^k$ conditioned on $(W_j^k)_{j = 1}^m$. Moreover, we denote $\hat{\mathbb{E}}^k$ as the sample average in $\mathcal{D}_k$. For example, $\hat{\mathbb{E}}^k[\phi(D)] = \sum_{i=1}^{n_k}\phi(D_{ik})/n_k$ where $D_{ki} \stackrel{\text{i.i.d.}}{\sim}\mathbb{P}^k$. We write $\text{Var}_{P}$ for the variance under $P$ and $\text{Var}_{\mathbb{P}^0}$ for the variance under $\mathbb{P}^0$. Specifically,  $\text{Var}_{\mathbb{P}^0}(\phi(D))$ is the variance of $\phi(D)$ for a random variable $D \in \mathcal{D}, D \sim \mathbb{P}^0$.

\begin{theorem}[Distributional CLT]\label{thm:perturbation-model}
Under the assumptions we mentioned above, for any Borel measurable square-integrable function $\phi_k: \mathcal{D}_k \xrightarrow[]{} \mathbb{R}$ for $k = 1, \dots, K$, we have
\begin{equation}
    \sqrt{m}\left(
   \begin{pmatrix}
    \hat{\mathbb{E}}^1[\phi_1(D)]\\
    \vdots \\
     \hat{\mathbb{E}}^K[\phi_K(D)]
    \end{pmatrix} - 
    \begin{pmatrix}
    \mathbb{E}^0[\phi_1(D)]\\
    \vdots \\
     \mathbb{E}^0[\phi_K(D)] 
    \end{pmatrix}\right) \xrightarrow[]{d}
     N\left(0, \Sigma^W  \odot \text{Var}_{\mathbb{P}^0}(\phi(D))\right)
\end{equation}
where  $\phi(D)^{\intercal} = (\phi_1(D), \dots,  \phi_K(D)) \in \mathbb{R}^{K}$, and $\Sigma^W\in \mathbb{R}^{K \times K}$ has 
$$(\Sigma^W)_{ij} = \frac{Cov_P(W^i, W^j)}{E[W^i]E[W^j]}.$$
Here, $\odot$ denotes element-wise multiplication (Hadamard product).
\end{theorem}

\begin{remark}[Random shift in $Y|X$]
As an example, let us study how conditional expectations shift under this model. We consider the case where the data consists of both covariates $X$ and target $Y$, that is $D = (X,Y)$. Let $A$ be some subset of the sample space of $X$ with $\mathbb{P}^0(X \in A)> 0$. Using a Taylor approximation,
\begin{equation*}
 \hat{\mathbb{E}}^1[Y|X \in A] - \hat{\mathbb{E}}^0[Y|X \in A] = \hat{\mathbb{E}}^1[\phi] - \hat{\mathbb{E}}^0[\phi] + o_P(1/\sqrt{m})
\end{equation*}
where $\phi = \frac{1_{X \in A} (Y - \mathbb{E}^0[Y|X \in A])}{\mathbb{P}^0[X \in A]}$. Thus, by the CLT,
\begin{equation*}
\sqrt{m}(\hat{\mathbb{E}}^1[Y|X \in A] - \hat{\mathbb{E}}^0[Y|X \in A]) \stackrel{d}{\rightarrow} \mathcal{N}(0,\Sigma_{11}^W \frac{\text{Var}_{\mathbb{P}^0}(Y | X \in A)}{\mathbb{P}^0(X \in A)} ).
\end{equation*}
Hence, the conditional expectation is randomly shifted based on three factors: 1) it is proportional to the distribution shift strength, $\sqrt{\Sigma_{11}^W/m}$; 2) it is proportional to the conditional noise,  $\sqrt{\text{Var}_{\mathbb{P}^0}(Y | X \in A)}$; and 3) it is inversely proportional to the probability $\sqrt{\mathbb{P}^0(X \in A)}$. This random shift in $Y | X$ needs to be accounted for both estimation and inference.
\end{remark}

\begin{remark}[Distributional covariance]
     We call $\Sigma^W/m$ the covariance between distributional perturbations. Two perturbed distributions $\mathbb{P}^{k}$ and $\mathbb{P}^{k'}$ are positively correlated if $\Sigma_{kk'}^W > 0$.  This means that when one perturbation slightly increases the probability of a certain event $A$, the other distribution tends to slightly increase the probability of this event as well. This can be seen by choosing the test function $\phi_k = 1_{D \in A}$:
\begin{equation*}
 \text{Cov}_P(\hat{\mathbb{P}}^k(D \in A), \hat{\mathbb{P}}^{k'}(D \in A)) =   \text{Cov}_P(\hat{\mathbb{E}}^k[1_{D \in A}], \hat{\mathbb{E}}^{k'}[1_{D \in A}]) \approx \frac{\Sigma_{k,k'}^W}{m} \text{Var}_{\mathbb{P}^0}(1_{D \in A}).
\end{equation*}
However, if two perturbed distributions are negatively correlated, a (random) increase in the probability of an event under one distribution would often coincide with a (random) decrease in the probability of this event under the other distribution, compared to $\mathbb{P}^0$.
\end{remark}

\begin{remark}[Estimation of variance]\label{remark:var}
    The variance under the target distribution $\text{Var}_{\mathbb{P}^0}(\phi(D))$ can be estimated via the empirical variance of $\phi$ on the pooled data $\{(D_{ki})_{i=1}^{n_k}\}_{k=1}^K$. Then, if $\phi$ has a finite fourth moment, we have $\widehat{\text{Var}}_{\mathbb{P}^0}(\phi(D)) = \text{Var}_{\mathbb{P}^0}(\phi(D)) + o_p(1)$. The details can be found in the Appendix \ref{sec:remark-2}.
\end{remark}

An extended version of Theorem~\ref{thm:perturbation-model} where $\phi_k: \mathcal{D}_k \xrightarrow[]{} \mathbb{R}^L$ for some $L \geq 1$ and the proof of Theorem~\ref{thm:perturbation-model} can be found in the Appendix, Section \ref{sec:proofs}. In this paper, we consider the target distribution as fixed. In the Appendix \ref{sec:remark}, we discuss a version of Theorem~\ref{thm:perturbation-model} that allows for a randomly shifted target distribution. In Section~\ref{sec:special-cases} of the Appendix, we discuss several special cases of the random perturbation model, including random shift across time, i.i.d.\ shifts, and random shifts between locations.

Theorem~\ref{thm:perturbation-model} tells us that $(\Sigma_{ij}^W)/m$ is the asymptotic covariance of $\hat{\mathbb{E}}^i[\phi(D)]$ and $\hat{\mathbb{E}}^j[\phi(D)]$ for \textit{any} Borel-measurable square-integrable function $\phi$ with unit variance under $\mathbb{P}^0$. 
Remarkably, under our random distributional perturbation model, a $(K + \binom{K}{2})$-dimensional parameter $\Sigma^W$ quantifies the similarity of multiple distributions and captures all possible correlations of all square-integrable functions.

\paragraph{The random shift model explains the linear patterns in the GTEx scatterplot.} Consider two perturbed distributions $\mathbb{P}^{1}$ and $\mathbb{P}^{2}$, and $L$ uncorrelated test functions with unit variances, $\phi_1(D), \dots, \phi_L(D)$. In the GTEx example, $\phi_{\ell}$ is a standardized product between gene expression levels of a randomly selected $\ell$-th gene pair. Moreover, in Appendix~\ref{sec:add-details}, we see that these $L$ test functions are approximately uncorrelated. From Theorem~\ref{thm:perturbation-model}, we have 
\begin{equation*}
     \sqrt{m}\left(
   \begin{pmatrix}
    \hat{\mathbb{E}}^1[\phi_{\ell}(D)]\\
     \hat{\mathbb{E}}^2[\phi_{\ell}(D)]
    \end{pmatrix} - 
    \begin{pmatrix}
    \mathbb{E}^0[\phi_{\ell}(D)]\\
     \mathbb{E}^0[\phi_{\ell}(D)] 
    \end{pmatrix}\right) \stackrel{d}{=}  \text{ }Z_{\ell} + o_p(1),
\end{equation*}
where
\begin{equation*}
     Z_{\ell} \stackrel{i.i.d}{\sim}
     N\left(0, \begin{pmatrix}
         \Sigma_{11}^{W} &  \Sigma_{12}^{W}\\
          \Sigma_{21}^{W} &  \Sigma_{22}^{W}
     \end{pmatrix}\right).
\end{equation*}
Therefore, $(
    \hat{\mathbb{E}}^1[\phi_{\ell}(D)] - \mathbb{E}^0[\phi_{\ell}(D)], 
     \hat{\mathbb{E}}^2[\phi_{\ell}(D)]- \mathbb{E}^0[\phi_{\ell}(D)])$ for $\ell = 1, \dots, L$, can be viewed as independent draws from a bivariate Gaussian distribution with distributional covariance $\Sigma^W/m$. This explains the linear patterns observed in the left plot of Figure~\ref{fig:gtex_pairs_plot}. Gaussianity can be evaluated by QQ plots of the statistic:
     \begin{equation} \label{eq:ttest}
          \left( \frac{1}{n_k} + \frac{1}{n_0} \right)^{-1/2} \left(\frac{\frac{1}{n_k}\sum_{i=1}^{n_k} \phi_{\ell}(D_{ki}) - \frac{1}{n_0}\sum_{i=1}^{n_0} \phi_{\ell}(D_{0i}) }
       {\hat{\text{sd}}_{\mathbb{P}^0}(\phi_{\ell}(D))} 
       \right).
    \end{equation}
    The QQ-plots of the statistic in \eqref{eq:ttest} are presented on the right-hand side of Figure~\ref{fig:gtex_pairs_plot}, when the target tissue is set to be brain cortex. These plots indicate that statistics in \eqref{eq:ttest} indeed follow a Gaussian distribution.

\section{Out-of-distribution generalization}\label{sec:distribution-generalization}

In this section, we demonstrate how, under the random distributional perturbation model, we can share information across datasets from different times and locations, and infer parameters of a partially observed target distribution. Specifically, we show how the distributional CLT from the previous section allows us to identify the optimal convex combination of source distributions to best approximate the target distribution.

We consider the following domain adaptation setting throughout the section. We aim to predict a target $Y$ based on some covariates $X$, that is the data is $D = (X,Y)$. We want to estimate $\theta^0 = \arg \min_\theta \mathbb{E}^0[\mathcal{L}(\theta,X,Y)]$ for some loss function $\mathcal{L}(\theta,X,Y)$. We consider the setting where we observe full $D = (X,Y)$ only on the source distributions $\mathbb{P}^1, \dots, \mathbb{P}^K$. More specifically, for the $k$-th source data ($k \in \{1, 2, \dots, K\}$), we observe $n_k$ i.i.d. samples $\mathcal{D}_k = \{D_{k1}, \dots, D_{kn_k}\}$ drawn from the source distribution $\mathbb{P}^k$. Furthermore, there are $n_0$ i.i.d. samples from the target distribution $\mathbb{P}^0$, but we only observe a subset $X$ and have $\mathcal{D}_0 = \{X_{01}, \dots, X_{0n_0}\}$.

 For any fixed non-negative weights $\beta_1,\ldots,\beta_K$ with $\sum_{k=1}^K \beta_k = 1$, one can consider weighted empirical risk minimization,
that is, one can set 
\begin{equation*}
   \hat \theta^\beta = \arg \min_\theta \sum_k \beta_k \hat{\mathbb{E}}^k[\mathcal{L}(\theta,X,Y)].
\end{equation*}
In the following, we want to study the excess risk
\begin{equation*}
    \mathbb{E}^0[\mathcal{L}(\hat \theta^\beta,X, Y)]  - \mathbb{E}^0[\mathcal{L}( \theta^0,X,Y)].
\end{equation*}
This will give us insights into how to choose the weights $\beta$ in an optimal fashion, and allow us to conduct statistical inference. 

\begin{lemma}[Out-of-distribution error of weighted ERM]\label{lemma:ooderror}
For each $\theta$ in an open subset of $\Omega$, let $\theta \mapsto \partial_\theta \mathcal{L}(\theta,x,y)$ be twice continuously differentiable in $\theta$ for every $x,y$. Assume that the matrix $\mathbb{E}^0[\partial_\theta^2 \mathcal{L}(\theta,X,Y)]$ exists and is nonsingular. Assume that the third partial derivatives of $\theta \mapsto \mathcal{L}(\theta,D)$ are dominated by a fixed function $h(D)$ for every $\theta$ in a neighborhood of $\theta^0$. We assume that $\partial_\theta \mathcal{L}(\theta^0,D)$, $\partial_\theta^2 \mathcal{L}(\theta^0,D)$ and $h(D)$ are square-integrable under $\mathbb{P}^0$. Assume that $\hat \theta^\beta - \theta^0 = o_P(1)$. Then,
\begin{equation*}
 m \left( \mathbb{E}^0[\mathcal{L}(\hat \theta^\beta,X,Y)]  - \mathbb{E}^0[\mathcal{L}( \theta^0,X,Y)] \right) \stackrel{d}{\rightarrow} \mathcal{E}
\end{equation*}
for some random variable $\mathcal{E}$, with 
\begin{equation*}
\mathbb{E}[\mathcal{E}] =  \beta^\intercal \Sigma^W \beta \cdot \mathrm{Trace}( \mathbb{E}^0[\partial_\theta^2 \mathcal{L}(\theta^0,X,Y)]^{-1} \mathrm{Var}_{\mathbb{P}^0}(  \partial_\theta \mathcal{L}(\theta^0,X,Y) )).
\end{equation*}
\end{lemma}

The proof of Lemma~\ref{lemma:ooderror} is provided in Appendix~\ref{sec:ood-error-proof}, along with the additional regularity assumptions under which the consistency assumption stated in Lemma~\ref{lemma:ooderror} holds. 

By this lemma, the optimal $\beta$ in terms of out-of-distribution error is
\begin{equation*}
    \beta^{*} =  \arg \min_{\substack{\beta:\beta^{\intercal} \mathbb{1} = 1\\\beta \geq 0} } \beta^\intercal \Sigma^W \beta.
\end{equation*}
A priori, $\beta^*$ is unknown to the researcher. In the following, we will describe estimation and inference for $\beta^*$. Our estimation strategy is based on the following observation. Using the distributional CLT, for any function $\phi(X)$ and any $\beta$,
\begin{equation}
  \sqrt{m}( \hat{\mathbb{E}}^0[\phi(X)] - \sum_{k=1}^K \beta_k \hat{\mathbb{E}}^0[\phi(X)] )  \stackrel{d}{\rightarrow} \mathcal{N}(0, \beta^\intercal \Sigma^W \beta \cdot \text{Var}_{\mathbb{P}^0}(\phi(X))).
\end{equation}
We can use this fact that the covariance term of the limiting distribution is inflated by the factor $\beta^{\intercal}\Sigma^W\beta$ regardless of the test functions to estimate $\beta^*$.  Consider different test functions of $X$, that is $\phi_{\ell}(X)$, $\ell = 1, \dots, L$.  For example, we may use individual covariates as test functions, $\phi_\ell = X_\ell$. Then, we can estimate $\beta^*$ by solving
\begin{equation}\label{eq:least-squares}
    \hat \beta = \argmin_{\substack{\beta:\beta^{\intercal} \mathbb{1} = 1 \\ \beta \geq 0}} \sum_{l=1}^L (\hat{\mathbb{E}}^0[\phi_l(X)] - \sum_k \beta_k \hat{\mathbb{E}}^k[\phi_l(X)]  )^2.
\end{equation}
and set the final estimate of $\theta^0$ as
\begin{equation*}
   \hat \theta = \arg \min_\theta \sum_k \hat \beta_k \hat{\mathbb{E}}^k[\mathcal{L}(\theta,X,Y)].
\end{equation*}
We will discuss the accuracy of this estimate in the following section. More specifically, we will derive confidence intervals for $\beta^*$ and $\theta^0$.

\begin{remark}
    In practice, one might wonder how to choose different test functions $\phi_l, l=1,\ldots,L$. 
    From a theoretical standpoint, under the distributional uncertainty model, one can technically choose any function with finite variance. 
    However, for actual applications, we suggest the following guideline. Our method considers a distribution $\mathbb{P}^k$ to be ``close" to $\mathbb{P}^0$ if the averaged values of test functions on $\mathcal{D}_k$ closely align with the corresponding averaged values on $\mathcal{D}_0$. Therefore, these test functions $\phi$ should be chosen to capture the dimensions of the problem where a close match is important. For example, in the numerical experiments in Section~\ref{sec:acs}, the target function of interest is the logarithmic value of income and using the means of covariates that are well-known to be related to the income variable as test functions demonstrates good performance. 
    We also provide simple diagnostics, presented in Figure~\ref{fig:gtex-diagnostic-plots} as standard residual plots and normal QQ-plots, that allow one to assess the fit of the distributional perturbation model and the choice of test functions.
\end{remark}

\subsection{Inference for $\beta^*$ and $\theta^0$}

In the following, we discuss how to form asymptotically valid confidence intervals for $\beta^*$ and $\theta^0$ when $\beta^*$ and $\hat{\beta}$ are allowed to have negative weights. As discussed earlier, we choose $L$ different test functions $\phi_1(X), \dots, \phi_L(X)$. For now we assume that test functions are uncorrelated and have unit variances under $\mathbb{P}^0$. Later we will discuss cases where test functions are correlated and have different variances.  In the previous section, we saw that $\hat \beta$ can be estimated via a least-squares problem with a linear constraint on the weights.
To remove the linear constraint in equation~\eqref{eq:least-squares}, we perform the following reparametrization,
\begin{align*}
    \hat \beta_{1:(K-1)}&= \argmin_{\beta} \sum_{l=1}^L \left((\hat{\mathbb{E}}^0[\phi_l(X)]-\hat{\mathbb{E}}^K[\phi_l(X)])  - \sum_{k=1}^{K-1} \beta_k (\hat{\mathbb{E}}^k[\phi_l(X)]- \hat{\mathbb{E}}^K[\phi_l(X)]) \right )^2.
\end{align*}
We can obtain $\hat \beta_K$ via $\hat \beta_K = 1- \sum_{k=1}^{K-1} \hat \beta_k$. We can re-write this as the following linear regression problem
\begin{equation*}
	 \hat \beta_{1:(K-1)} = \argmin_{\beta_{1:(K-1)} \in \mathbb{R}^{K-1}}  \| \Tilde{\Phi}^0 - \Tilde{\Phi}\beta_{1:(K-1)} \|_2^2,
\end{equation*}
where the feature matrix in the linear regression is defined as 
\begin{equation*}
    \Tilde{\Phi} = \begin{pmatrix}
        \hat{\mathbb{E}}^1[\phi_{1}]-\hat{\mathbb{E}}^K[\phi_{1}]  & \dots & \hat{\mathbb{E}}^{K-1}[\phi_{1}] -\hat{\mathbb{E}}^K[\phi_{1}]\\
        \vdots & & \vdots \\
        \hat{\mathbb{E}}^1[\phi_{L}] -\hat{\mathbb{E}}^K[\phi_{L}] & \dots & \hat{\mathbb{E}}^{K-1}[\phi_{L}]-\hat{\mathbb{E}}^K[\phi_{L}] 
    \end{pmatrix} \in \mathbb{R}^{L \times (K-1)},
\end{equation*}
and the outcome vector in the linear regression is defined as
\begin{equation*}
    \Tilde{\Phi}^0 = \left( \hat{\mathbb{E}}^0[\phi_{1}]-\hat{\mathbb{E}}^K[\phi_{1}] , \dots, \hat{\mathbb{E}}^0[\phi_{L}] -\hat{\mathbb{E}}^K[\phi_{L}]\right)^{\intercal} \in \mathbb{R}^{L \times 1}.
\end{equation*}
Since test functions were assumed to be uncorrelated and have unit variances under $\mathbb{P}_0$, using Theorem~\ref{thm:perturbation-model}, rows of the feature matrix $\Tilde{ \Phi}$ and the outcome vector $\Tilde{ \Phi}^0$ are asymptotically i.i.d., with each row drawn from a centered Gaussian distribution. Therefore, the problem may be viewed as a standard multiple linear regression problem where the variance of the coefficient $\hat{\beta}_{1:(K-1)}$  is estimated as
\begin{equation*}
    \widehat{\text{Var}}(\hat{\beta}_{1:(K-1)}) = (\Tilde{\Phi}^{\intercal}\Tilde{\Phi})^{-1} \hat{\sigma}^2,
\end{equation*}
and the residual variance is calculated as
\begin{equation*}
    \hat{\sigma}^2 = \frac{1}{L - K + 1} \| \Tilde{\Phi}^0 - \Tilde{\Phi}\hat\beta \|_2^2.
\end{equation*}
Different from standard linear regression, in our setting each column represents a data distribution and each row represents a test function, and Gaussianity does not hold exactly in finite samples. Therefore, a priori it is unclear whether standard statistical tests in linear regression (such as $t$-test and $F$-test) carry over to the distributional case. The following theorem tells us that the standard linear regression results still hold. Therefore, we can conduct statistical inference for the optimal weights $\beta^*$ in the same manner as we conduct t-tests and F-tests in a standard linear regression. The proof of Theorem~\ref{theorem:dlm} can be found in Appendix, Section~\ref{sec:dlmproofs}. Note that these results assume a finite number of test functions and datasets but are asymptotic in terms of $m(n)$ as they rely on the distributional CLT.

\begin{theorem}\label{theorem:dlm}
    Assume that test functions $\phi_1, \dots, \phi_L$ are uncorrelated and have unit variances under $\mathbb{P}^0$. Let $\Tilde{\Phi}$ and $\hat{\sigma}^2$ be defined as above. Then, we have
     \begin{equation*}
     \left((\Tilde{\Phi}^{\intercal}\Tilde{\Phi})^{-1}\hat{\sigma}^2\right)^{-\frac{1}{2}} \left(\hat{\beta}_{1:(K-1)} - \beta^*_{1:(K-1)}\right) \stackrel{d}{=} t_{K-1}(L - K + 1) +o_p(1),
     \end{equation*}
    where $t_{p}(d)$ follows the $p$-dimensional multivariate t-distribution with $d$ degrees of freedom. 
\end{theorem}

\begin{remark}[Distributional t-test]
	For this test, the null hypothesis is that the dataset $\mathcal{D}_i$ should have weight zero in the weighted empirical risk minimization. Mathematically, this corresponds to testing the null hypothesis $\beta_i^* = 0$. Thus, we consider the following hypotheses:
	\begin{equation*}
		H_0: \beta_i^* = 0, \quad H_1: \beta_i^* \neq 0.
	\end{equation*}
	The test statistic is defined as
	\begin{equation*}
		t = \frac{\hat{\beta}_i}{\sqrt{\left( \left(\Tilde{\Phi}^{\intercal}\Tilde{\Phi}\right)^{-1}\hat{\sigma}^2\right)_{ii}}}.
	\end{equation*}
	From Theorem \ref{theorem:dlm}, under the null hypothesis,
	\begin{equation*}
		t \sim t(L-K+1),
	\end{equation*}
	where $t(L-K+1)$ is a univariate $t$-distribution with $L-K+1$ degrees of freedom.
\end{remark}

\begin{remark}[Distributional F-test] For this test, the null hypothesis is  that each dataset is equally informative for the target distribution. Mathematically, this corresponds to testing the null hypothesis $\beta_i^* = 1/K$ for $i = 1,\ldots,K$. Thus, we consider the following hypotheses:
\begin{align*}
	&H_0: \, \,  \beta_1^* =  \ldots = \beta_K^* =\frac{1}{K}, \\
	&H_1: \, \,  \beta_i^* \neq \frac{1}{K}, \text{ for at least one }i.
\end{align*}
The test statistic is defined as
\begin{equation*}
	F =   \left(\hat{\beta}_{1:(K-1)} - \frac{1}{K} \cdot \bm{1}_{K-1} \right)^{\intercal}\left((\Tilde{\Phi}^{\intercal}\Tilde{\Phi})^{-1}\hat{\sigma}^2\right)^{-1} \left(\hat{\beta}_{1:(K-1)} - \frac{1}{K} \cdot \bm{1}_{K-1}\right)/(K-1).
\end{equation*}
From Theorem \ref{theorem:dlm}, under the null hypothesis, 
\begin{equation*}
	F \sim F_{K-1, L-K+1},
\end{equation*}
where $F_{K-1, L-K+1}$ is a F-distribution with degrees of freedom $K-1$ and $L-K+1$.
   
\end{remark}

\begin{remark}[Distributional confidence intervals for parameters of the target distribution]\label{remark:ci-target}
 We will now discuss how to form asymptotically valid confidence intervals for $\theta^0 = \arg \min_\theta \mathbb{E}^0[\mathcal{L}(\theta,X,Y)]$ when $m \xrightarrow[]{} \infty$ and $L \xrightarrow[]{} \infty$. Here we assume that the influence function $\phi = - \mathbb{E}^0[\partial_\theta^2 \mathcal{L}(\theta^0,X,Y)]^{-1}  \partial_\theta \mathcal{L}(\theta^0,X,Y)$ has a finite fourth moment under $\mathbb{P}^0$. Then, we can form $(1-\alpha)$-confidence intervals for $\theta^0$ with $\hat{\theta} = \arg \min_\theta \sum_k \hat \beta_k \hat{\mathbb{E}}^k[\mathcal{L}(\theta,X,Y)]$ as follows:
\begin{equation*}
    \hat \theta \pm z_{1-\alpha/2} \cdot \sqrt{\widehat{\text{Var}}_{\mathbb{P}^0}(\phi)} \cdot \sqrt{ \frac{1}{L}\sum_{\ell = 1}^{L}  \left(\hat{\mathbb{E}}^0[\phi_{\ell}] - \sum_{k=1}^{K} \hat{\beta}_k \hat{\mathbb{E}}^k[\phi_{\ell}] \right)^2 }.
\end{equation*}
 Here, $\widehat{\text{Var}}_{\mathbb{P}^0}(\phi)$ denotes the empirical variance of $$\hat \phi(X,Y) = - ( \sum_{k=1}^K \hat \beta_k \hat{\mathbb{E}}^k[\partial_\theta^2 \mathcal{L}( \hat \theta,X,Y)])^{-1}  \partial_\theta \mathcal{L}(\hat \theta,X,Y)$$ on the pooled donor data $\{\{X_{ki},Y_{ki}\}_{i=1}^{n_k}\}_{k=1}^K$. The proof of this result can be found in the Appendix, Section~\ref{sec:prediction-intervals}.
\end{remark}

\begin{remark}[Distributional sample size and degrees of freedom]
    Note that the asymptotic variance of ``distributional" regression parameters $(\tilde \Phi^\intercal \tilde \Phi)^{-1} \hat \sigma^2$ has the same algebraic form as the asymptotic variance for regression parameters in a standard Gaussian linear model. In the classical setting, the test statistics of regression parameters follow $t$-distributions with $n-p+1$ degrees of freedom, where $n$ is the sample size and $p$ is the number of covariates. In our setting, we have $L-K+1$ degrees of freedom, where $L$ corresponds to the number of test functions and $K$ corresponds to the number of donor distributions.
\end{remark}

\begin{remark}[Correlated test functions]\label{remark:correlated} We assumed that the $(\phi_1(X), \dots, \phi_L(X))$ are uncorrelated and  have unit variances under $\mathbb{P}^0$. In practice, this might not be the case. In such cases, we can apply a linear transformation $T$ to the test functions in a pre-processing step to obtain uncorrelated test functions with unit variances. We define the transformation matrix $T = (\hat{\Sigma}^{\Phi})^{-1/2}$, where $\hat{\Sigma}^{\Phi}$ is an estimate of the covariance matrix $\Sigma^{\Phi}$ of $(\phi_1(X), \dots, \phi_L(X))$ on the pooled data. From Remark~\ref{remark:var}, if $(\phi_1(X), \dots, \phi_L(X))$ have finite fourth moments, we have $\hat{\Sigma}^{\Phi} =\Sigma^{\Phi} + o_p(1)$. 
\end{remark}

\section{ACS Income Data}\label{sec:acs}

In this section, we demonstrate the effectiveness and robustness of our method when the training dataset is obtained from multiple sources. We focus on the ACS Income dataset \citep{ding2021retiring} where the goal is to predict the logarithmic value of individual income based on tabular census data.

Similar as in \citet{shen2023mo}, we consider a scenario where we initially have a limited training dataset from California (CA), and subsequently, we obtain additional training dataset from Puerto Rico (PR). We iterate through the target state among the rest of 49 states. Note that there are substantial economic, job market, and cost-of-living disparities between CA and PR. Most of the states are more similar to CA than PR. Therefore, increasing the training dataset with data sourced from PR will lead to substantial performance degradation of the prediction algorithm on the target state.

We compare three methods in our study. In the first method, we fit the XGBoost on the training dataset in the usual manner, with all samples assigned equal weights. In the second method, we employ sample-specific importance weights. The importance weights are calculated as follows. First, pool the covariates $X$ of the training dataset (CA + PR) and the target dataset. Let $A_i$ be the label which is 1 if the $i$-th sample is from the target dataset and 0 otherwise. Then, the importance weight for the $i$-th sample is obtained as
\begin{equation*}
    \hat{w}_{IW, i} = \frac{\hat{P}(X = X_i|A_i = 1)}{\hat{P}(X = X_i|A_i= 0)} = \frac{\hat{P}(A_i = 1|X = X_i)\hat{P}(A_i = 0)}{\hat{P}(A_i = 0|X = X_i)\hat{P}(A_i = 1)},
\end{equation*}
where we estimate $\hat{P}(A_i|X = X_i)$ using XGBoost on the pooled covariates data and estimate $\hat{P}(A_i=1)$ as the ratio of the sample size of the target data to the sample size of the pooled data.  

In the third method, based on our approach, the XGBoost is fitted on the training dataset, but this time samples receive different weights depending on whether they originate from CA or PR. We utilize the occupation-code feature to construct test functions. The test function $\phi_{\ell}(D_i)$ is defined as whether $i$-th sample has the $\ell$-th occupation code. The occupation code is a categorical variable representing around 500 different occupation categories. The data dictionary at ACS PUMS documentation\footnote{https://www.census.gov/programs-surveys/acs/microdata/documentation.2018.html} provides the full list of occupation codes. The number of test functions in total is therefore around 500. Then, we estimate distribution-specific-weights $(w_{\text{CA}}, w_{\text{PR}})$ by running
 \begin{equation*}
     \hat{w}_{\text{CA}}, \hat{w}_{\text{PR}} = \argmin_{ w_{\text{CA}}+w_{\text{PR}} = 1} \sum_{\ell = 1}^{L} (\hat{\mathbb{E}}^{\text{Target}}[\phi_{\ell}] - w_{\text{CA}}\hat{\mathbb{E}}^{\text{CA}}[\phi_{\ell}] - w_{\text{PR}}\hat{\mathbb{E}}^{\text{PR}}[\phi_{\ell}])^2.
 \end{equation*}
 Note that $\hat{\mathbb{E}}^{\text{state}}[\phi_{\ell}]$ is the proportion of samples of the srtate with the $\ell$-th occupation code. 

The results with the target states of Texas (TX) and Washington (WA) are given in Figure \ref{fig:ACS}. Similar patterns are observed for other target states, aligning with either TX or WA behaviors. Results for other target states can be found in the Appendix~\ref{sec:add-details}. 

For the target TX, as we include the PR data, the Mean Squared Error (MSE) initially drops for all the methods. However, as more data from PR is added, the MSE starts to increase for both methods with equal weights and importance weights. In contrast, our method demonstrates robustness, with the MSE decreasing after including PR data and staying consistently low, even when the PR source becomes dominant in the training dataset. Turning to the target WA, for both methods with equal weights and importance weights, adding data sourced from PR results in a straightforward increase in MSE. In contrast, our method assigns weights close to 0 for samples originating from PR, preventing significant performance degradation.

\section{GTEx Data}\label{sec:gtex}

In this section, we illustrate our methods using the GTEx data. We will also run model diagnostics to evaluate model fit. The GTEx V6 dataset provides RNA-seq gene-expression levels collected from 450 donors across 44 tissues. 

Treating each tissue as a separate study, we have 44 different datasets. As shown in Figure \ref{fig:gtex_pairs_plot}, some tissues exhibit higher correlations than others. To demonstrate our method, we select 5 different tissues (same as in Figure~\ref{fig:gtex_pairs_plot}): Adipose subcutaneous, Adipose visceral omentum, Brain cortex, Brain frontal cortex, and Brain cerebellum. Let brain cortex be our target tissue (the third tissue in Figure~\ref{fig:gtex_pairs_plot}). The code and the data are available as R package \texttt{dlm} at \url{https://github.com/yujinj/dlm}. 

For our test functions, we randomly sample 1000 gene-pairs and define a test function as the standardized product of gene-expression levels of a pair. In total, we have 1000 test functions. Most genes are known to be uncorrelated with each other, with the 90th percentile range of correlations between two different test functions being about $[-0.06, 0.06]$.
We employ group Lasso to estimate the sparse inverse covariance matrix. As the estimated inverse covariance matrix is close to a diagonal matrix (details provided in the Appendix~\ref{sec:add-details}), we do not perform a whitening transformation.

The R package \texttt{dlm} (\url{https://github.com/yujinj/dlm}) contains a function \texttt{dlm} that performs distributional linear regressions:
\begin{verbatim}
    dlm(formula, test.function, data, whitening)
\end{verbatim}
Here, the formula parameter specifies which model we want to fit. 
The test.function parameter represents considered test functions.
The data parameter passes a list of datasets.
The whitening parameter is a boolean indicating whether one wants to perform a whitening transformation. Additional details and descriptions of the function can be found on the GitHub page.

Below is the summary output of the \texttt{dlm} function with a model formula ``Brain cortex $\sim$ Adipose subcutaneuous + Adipose visceral omentum + Brain frontal cortex + Brain cerebellum" with 1000 test functions defined as standardized products of randomly selected gene-pairs. Note that the summary output of the \texttt{dlm} function closely resembles that of the commonly used \texttt{lm} function. 

\begin{Verbatim}
Call:
dlm(formula = Brain_Cortex ~ 
            Adipose_Subcutaneous + 
            Adipose_Visceral_Omentum + 
            Brain_Frontal_Cortex_BA9 +
            Brain_Cerebellum, 
    test.function = phi, 
    data = GTEx_data, whitening = FALSE)

Residuals:
     Min       1Q   Median       3Q      Max 
-0.56838 -0.09436 -0.00604  0.08739  0.40917 

Coefficients:
                           Estimate Std. Error t value Pr(>|t|)    
Adipose_Subcutaneous     -0.0001295  0.0304036  -0.004   0.9966    
Adipose_Visceral_Omentum  0.0462641  0.0250888   1.844   0.0655 .  
Brain_Frontal_Cortex_BA9  0.7846112  0.0184325  42.567  < 2e-16 ***
Brain_Cerebellum          0.1692543  0.0217128   7.795 1.61e-14 ***
---
Signif. codes:  0 ‘***’ 0.001 ‘**’ 0.01 ‘*’ 0.05 ‘.’ 0.1 ‘ ’ 1

Residual standard error: 0.1377 on 997 degrees of freedom
Multiple R-squared:  0.5088,	Adjusted R-squared:  0.5073 
F-statistic: 344.3 on 3 and 997 DF,  p-value: < 2.2e-16
\end{Verbatim}

The coefficient represents the estimated weight for each perturbed dataset. The sum of coefficients equals to one. The estimated weight for the brain-frontal-cortex dataset is close to 1: 0.7846112. On the other hand, the estimated weights for adipose-subcutaneous and adipose-visceral-omentum datasets are close to 0: -0.0001295 and 0.0462641. This supports the observations made in Figure~\ref{fig:gtex_pairs_plot}, where the brain-cortex dataset (target) exhibits a higher correlation with the brain-frontal-cortex dataset than others. 

Let us discuss the summary output in more detail. All tests have in common that the probability statements are with respect to distributional uncertainty, that is, the uncertainty induced by the random distributional perturbations.

\begin{enumerate}
    \item (Distributional t-test). The summary output provides a t-statistic and a p-value for each estimated weight. Based on the output, the estimated weight for the brain-frontal-cortex dataset is highly significant (with a t-statistic 42.567 and a p-value less than $2\cdot10^{-16}$). We have enough evidence to believe that the true optimal weight for the brain-frontal-cortex dataset is nonzero, and thus we conclude that the brain-frontal-cortex dataset is important for predicting the target dataset.
    On the other hand, the estimated weights for the adipose-subcutaneous and adipose-visceral-omentum datasets are not statistically significant (with t-statistics -0.004 and 1.844 and p-values 0.9966 and 0.0655). We conclude those datasets are not as important as the brain-frontal-cortex dataset for predicting the target dataset.  
    \item (Distributional F-test). The summary output provides a F-statistic and a corresponding p-value for a null hypothesis: the optimal weights are uniform ($\beta^* = (0.25, 0.25, 0.25, 0.25)$ in our example). This means that different datasets are equally informative in explaining the target dataset. The F-statistic is given as 344.3, which follows the F-distribution with degrees of freedom 3 and 997 under the null hypothesis. The resulting p-value is less than $2.2\cdot 10^{-16}$. Therefore, we reject the null hypothesis. From Figure~\ref{fig:gtex_pairs_plot} and the estimated weights, we can see that the brain-frontal-cortex and brain-cerebellum datasets are more informative in explaining the target dataset than the adipose-subcutaneous and adipose-visceral-omentum datasets.
    
    \item (R-squared). 
    Usually, the coefficient of determination (R-squared) measures how much of the variation of a target variable is explained by other explanatory variables in a regression model. In our output, R-squared measures how much of the unexplained variance of a target dataset, compared to considering uniform weights, is explained by considering our proposed weights. Specifically, it is calculated as
    \begin{equation*}
        R^2 = 1 - \frac{\text{RSS with estimated weights} }{\text{RSS with uniform weights}}.
    \end{equation*}
    R-squared is close to 0 if the estimated weights are close to uniform weights. 
\end{enumerate}
Finally, we should consider diagnostic plots to evaluate the appropriateness of the fitted distribution shift model. This can be evaluated using a distributional Tukey-Anscombe and a distributional QQ-Plot. In contrast to conventional residual plots and QQ-Plots, each point in these plot corresponds to a mean of a test function instead of a single observation. Figure~\ref{fig:gtex-diagnostic-plots} shows that neither the QQ-Plot nor the TA plot indicate a deviation from the modelling assumptions.

\begin{figure}[t!]
    \centering
    \includegraphics[scale = 0.5]{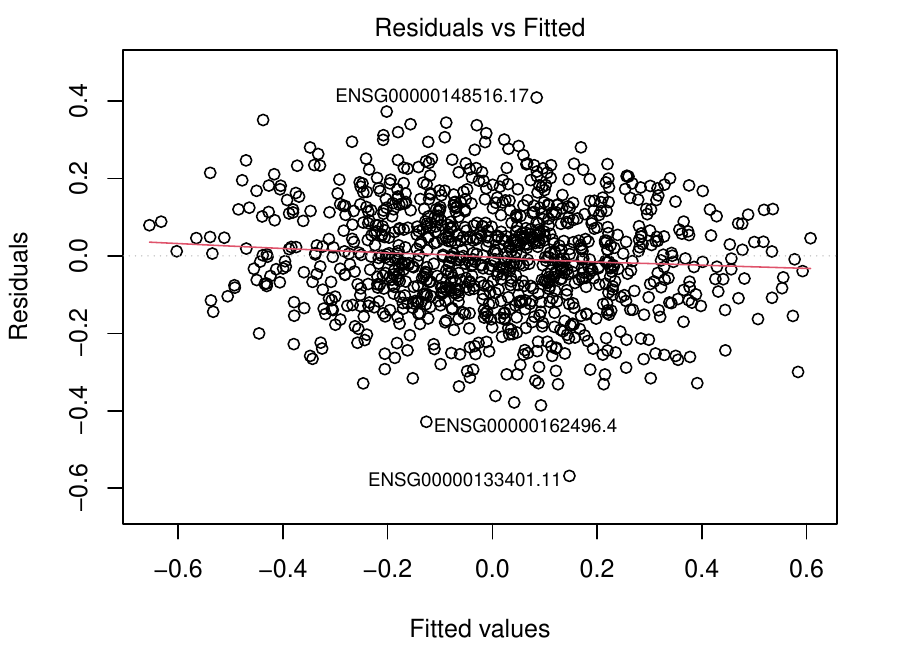}
    \includegraphics[scale = 0.5]{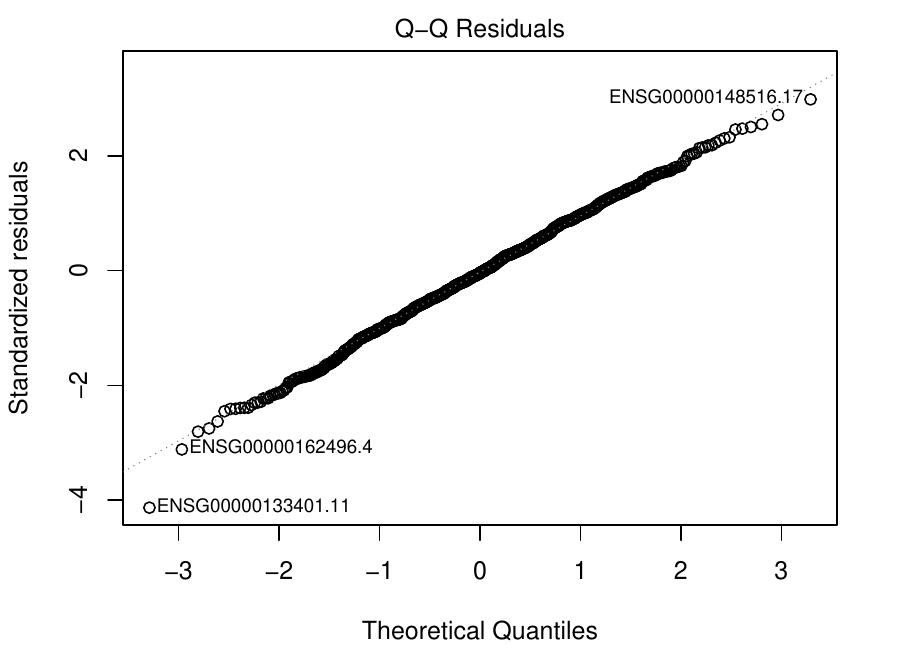}
    \caption{The distributional residual plot and distributional QQ-Plot for the GTEx data. In contrast to conventional residual plots and QQ-Plots, each point corresponds to a mean of a test function instead of a single observation.}
    \label{fig:gtex-diagnostic-plots}
\end{figure}

\section{Discussion}\label{sec:discussion}

In many practical settings, there is a distribution shift between the training and the target distribution.
Existing distribution shift models fall into two categories. One category, which we term ``sparse distribution shift", assumes that the shift affects specific parts of the data generation process while leaving other parts invariant. Existing methods attempt to re-weight training samples to match the target distribution or learn invariant representations of features. The second category considers worst-case distributional shifts. More specifically, they consider the worst-case behavior of a statistical functional over a fixed neighborhood of the model based on the distributional distance. This often requires knowledge of the strength and shape of the perturbations.

In contrast, we consider random dense distributional shifts, where the shift arises as the superposition of many small random changes. We have found that the random perturbation model can be useful in modeling empirical phenomena observed in Figure~\ref{fig:gtex_pairs_plot}. Moreover, we see that even when the overlap assumptions seem to be violated (when the reweighing approach fails), or when the distribution shift appears large in Kullback-Leibler divergence, the random perturbation model may still be appropriate and useful. Under the random distributional perturbation model, we establish foundations for transfer learning and generalization. 

Our method shares methodological similarities with synthetic controls. While synthetic controls are typically applied to panel data under the assumption of a linear factor model, our procedure can be applied to model distribution shifts of any type of data (discrete, continuous, ordinal) with or without time structure. Furthermore, our method justifies the linearity in synthetic control methods under random distributional shifts, rather than assuming a linear factor model. The random distributional shift assumption can be evaluated based on plots in Figure~\ref{fig:gtex-diagnostic-plots}. Additionally, we propose a generalization of our procedure to empirical risk minimization. 

In practice, perturbations may involve a combination of sparse and dense shifts. In such cases, hybrid approaches (sparse + dense) may be appropriate. We view hybrid models as an important generalization and leave it for future work.

A companion R package, \texttt{dlm}, is available at \url{https://github.com/yujinj/dlm}. Our package performs distributional generalization under the random perturbation model. Users can replicate the results in Section \ref{sec:gtex} with the data and code provided in the package.

\section{Acknowledgments}

The authors are grateful for insightful discussions with Naoki Egami, Kevin Guo, Ying Jin, Hongseok Namkoong, and Elizabeth Tipton. This work was supported by Stanford University’s Human-Centered Artificial Intelligence (HAI) Hoffman-Yee Grant, and by the Dieter Schwarz Foundation.

\newpage 

\bibliography{bibliography}

\newpage

\section*{Appendix}

\section{Proof of Theorem~\ref{thm:perturbation-model}}\label{sec:proofs}

\subsection{Extended Theorem~\ref{thm:perturbation-model}}\label{sec:extended-lemma}

In the following, we present the extended version of Theorem~\ref{thm:perturbation-model}. Note that we get Theorem~\ref{thm:perturbation-model} in the main text by defining $\phi(D) = (\phi_1(D), \dots, \phi_K(D))^{\intercal}$.

\begin{extended-theorem}\label{extended-theorem}
Under the assumptions of Theorem~\ref{thm:perturbation-model}, for any Borel measurable square-integrable function $\phi: \mathcal{D} \xrightarrow[]{} \mathbb{R}^L$, we have
\begin{equation}
    \sqrt{m}\left(
   \begin{pmatrix}
    \frac{1}{n_1}\sum_{i=1}^{n_1}\phi(D_{1i})\\
    \vdots \\
     \frac{1}{n_K}\sum_{i=1}^{n_K}\phi(D_{Ki})
    \end{pmatrix} - 
    \begin{pmatrix}
    \mathbb{E}^0[\phi(D)]\\
    \vdots \\
     \mathbb{E}^0[\phi(D)] 
    \end{pmatrix}\right) \xrightarrow[]{d}
     N\left(0, \Sigma^W \otimes \text{Var}_{\mathbb{P}^0}({\phi}(D))\right),
\end{equation}
where $\Sigma^W \in \mathbb{R}^{K \times K}$ is
$$(\Sigma^W)_{ij} = \frac{Cov(W^i, W^j)}{E[W^i]E[W^j]}.$$
Here, $\otimes$ denotes a kronecker product. Note that $\Sigma^W \otimes \text{Var}_{\mathbb{P}^0}({\phi}(D))$ can be written as
\begin{align*}
\setlength{\arraycolsep}{8pt}
\renewcommand{\arraystretch}{1.2}
\left[
\begin{array}{c@{}c:c:c:c}
      & \Sigma^W_{11} \cdot \text{Var}_{\mathbb{P}^0}( {\phi}(D))   &  \Sigma^W_{12} \cdot \text{Var}_{\mathbb{P}^0}( {\phi}(D))       & \cdots &    
    \Sigma^W_{1K} \cdot \text{Var}_{\mathbb{P}^0}( {\phi}(D))                \\ \hdashline
        & \Sigma^W_{21} \cdot \text{Var}_{\mathbb{P}^0}( {\phi}(D))    & \Sigma^W_{22} \cdot \text{Var}_{\mathbb{P}^0}( {\phi}(D))     & \cdots & \Sigma^W_{2K} \cdot \text{Var}_{\mathbb{P}^0}( {\phi}(D))         \\ \cdashline{1-5}
       &{\vdots}      & {\vdots}        &  &{\vdots}                 \\ \cdashline{1-5}
   & \Sigma^W_{K1} \cdot \text{Var}_{\mathbb{P}^0}( {\phi}(D))       & \Sigma^W_{K2} \cdot \text{Var}_{\mathbb{P}^0}( {\phi}(D))        & \cdots     &  \Sigma^W_{KK} \cdot \text{Var}_{\mathbb{P}^0}( {\phi}(D))      \\
\end{array}\right].
\end{align*}
\end{extended-theorem}

\quad\\
In the following, we prove Theorem~\ref{thm:perturbation-model}. From the proof, we can get Extended Theorem~\ref{extended-theorem} by considering 
\begin{equation*}
    (\hat{\mathbb{E}}^1[\phi_1(D)],\dots, \hat{\mathbb{E}}^1[\phi_L(D)], \hat{\mathbb{E}}^2[\phi_1(D)], \dots, \hat{\mathbb{E}}^2[\phi_L(D)], \dots \dots, \hat{\mathbb{E}}^K[\phi_1(D)], \dots, \hat{\mathbb{E}}^K[\phi_L(D)]),
\end{equation*}
where $\phi(D) = (\phi_1(D), \dots, \phi_L(D))^{\intercal} \in \mathbb{R}^L$.

\subsection{Auxiliary Lemma for Theorem~\ref{thm:perturbation-model}}

Let us first state an auxiliary lemma that will turn out helpful for proving Theorem~\ref{thm:perturbation-model}.

\begin{lemma}\label{lemma:pert_asymp}
Let the assumptions of Theorem~\ref{thm:perturbation-model} hold. For any bounded measurable $\phi_1, \dots, \phi_K: \mathcal{D} \xrightarrow[]{} \mathbb{R}$, we have that
\begin{equation*}
   \sqrt{m}\left(
   \begin{pmatrix}
    \mathbb{E}^{1}[\phi_1(D)]\\
    \vdots \\
     \mathbb{E}^{k}[\phi_K(D)]
    \end{pmatrix} - 
    \begin{pmatrix}
    \mathbb{E}^0[\phi_1(D)]\\
    \vdots \\
     \mathbb{E}^0[\phi_K(D)] 
    \end{pmatrix}\right)\xrightarrow{d}
     N\left(0, \Sigma^W \odot \text{Var}(\overrightarrow{\phi}(D))\right)
\end{equation*}
where $(\Sigma^W)_{ij} = \frac{Cov(W^i, W^j)}{E[W^i]E[W^j]}$ and $\overrightarrow{\phi}(D) = (\phi_1(D), \dots, \phi_K(D))^{\intercal}$.
\end{lemma}
\begin{proof}
Let $\psi_{k} = \phi_{k} \circ h$.
Without loss of generality, assume that $\mathbb{E}^0[\psi_{k}(U)] = 0$ for $k = 1, \dots, K$. Note that
\begin{equation*}
    \sqrt{m}(\mathbb{E}^{k}[\psi_{k}(U)] - \mathbb{E}^0[\psi_{k}(U)]) = \frac{\sqrt{m}\sum_{j=1}^{m}\int_{x \in I_j} \psi_{k}(x) dx \cdot (W^k_j - E[W^k])}{\sum_{j=1}^{m}W^k_j/m}. 
\end{equation*}
Let 
\begin{equation*}
    Y_{m,j} := \sqrt{m} \begin{pmatrix}
    \int_{x \in I_j} \psi_1(x) dx \cdot (W^1_j - E[W^1]) \\
    \vdots \\
    \int_{x \in I_j} \psi_K(x) dx \cdot (W^K_j - E[W^K])
    \end{pmatrix}.
\end{equation*}
First, note that 
\begin{equation}\label{eq:cltcond1}
    E[Y_{m,j}] = 0
\end{equation}
for all $j$. As the second step, we want to show that 
\begin{equation}\label{eq:cltcond2}
    \sum_{j=1}^{m}\text{Cov}(Y_{m,j})\xrightarrow[]{} \text{Var}(\overrightarrow{W})\odot \text{Var}_{\mathbb{P}^0}(\overrightarrow{\psi}(U)),
\end{equation}
where $\overrightarrow{W} = (W^1, \dots, W^K)^{\intercal}$, $\overrightarrow{\psi}(U) = (\psi_1(U), \dots, \psi_K(U))^{\intercal}$, and $\odot$ is an element-wise multiplication. For any $f \in L^2([0,1])$, define $\Pi_m(f)$ as
\begin{equation*}
    \Pi_m(f)(x) = \sum_{j=1}^{m} \left(m\int_{x \in I_j} f(x) dx\right) \cdot I(x \in I_j).
\end{equation*}
Then, we have
\begin{align*}
    & \Bigg| \left(\sum_{j=1}^{m}\text{Cov}(Y_{m,j}) -\text{Var}(\overrightarrow{W})\odot \text{Var}_{\mathbb{P}^0}(\overrightarrow{\psi}(U))\right)_{kk'}\Bigg| \\
    & \quad\quad = \Bigg| \left(m \sum_{j=1}^{m} \int_{x \in I_j} \psi_k(x) dx \int_{x \in I_j}\psi_{k'}(x)dx - \sum_{j=1}^{m}\int_{x \in I_j}\psi_k(x)\psi_{k'}(x)dx \right) \cdot \text{Cov}(W^k, W^{k'})\Bigg|  \\
    & \quad\quad \leq ||\psi_k||_2||\psi_{k'} - \Pi_m(\psi_{k'})||_2 \cdot |\text{Cov}(W^k, W^{k'})|\xrightarrow[]{} 0
\end{align*}
for $1 \leq k, k' \leq K$ as $m$ goes to infinity. This is because any bounded function can be approximated by a sequence of step functions of the form $\sum_{j=1}^{m}b_jI(x\in I_j)$.
Next we will show that for any $\epsilon > 0$,
\begin{equation}\label{eq:cltcond3}
    g_m(\epsilon) = \sum_{j=1}^{m}E[||Y_{m,j}||_2^2 ; ||Y_{m,j}||_2 \geq \epsilon] \xrightarrow[]{}0. 
\end{equation}
Let $||\psi_1||_{\infty},\dots, ||\psi_K||_{\infty} \leq B$. Then this is implied by the dominated convergence theorem as
\begin{align}
\sum_{j=1}^{m}E[||Y_{m,j}||_2^2 ; ||Y_{m,j}||_2 \geq \epsilon ] & \leq  B^2 E[||\overrightarrow{W}- E[\overrightarrow{W}]||_2^2 I(B||\overrightarrow{W}- E[\overrightarrow{W}]||_2/\sqrt{m} \geq \epsilon)]  \xrightarrow[]{} 0.
\end{align}
Combining equations~\eqref{eq:cltcond1}, \eqref{eq:cltcond2}, and \eqref{eq:cltcond3}, we can apply multivariate Lindeberg's CLT.  Together with Slutsky's theorem, we have
\begin{align*}
   \sqrt{m}\left(\begin{pmatrix}
    \mathbb{E}^{1}[\phi_1(D]\\
    \vdots \\
     \mathbb{E}^{k}[\phi_K(D)]
    \end{pmatrix} - 
    \begin{pmatrix}
    \mathbb{E}^0[\phi_1(D]\\
    \vdots \\
     \mathbb{E}^0[\phi_K(D)] 
    \end{pmatrix}\right) & = \sum_{j=1}^m Y_{m,j} \odot \begin{pmatrix}
        1/(\sum_{j=1}^m W_j^1/m)\\
        \vdots \\
        1/(\sum_{j=1}^m W_j^K/m)
    \end{pmatrix}
    \\
    & \xrightarrow[]{d}
     N\left(0, \Sigma^W \odot \text{Var}_{\mathbb{P}^0}(\overrightarrow{\phi}(D))\right)
\end{align*}
where $\Sigma^W_{ij} = \frac{Cov(W^i, W^j)}{E[W^i]E[W^j]}$. This completes the proof. 
\end{proof}

\subsection{Proof of Theorem~\ref{thm:perturbation-model}}

\begin{proof}
    For any $\phi_k \in L^2(\mathbb{P})$ and for any given $\epsilon >0$, there exits a bounded function $\phi_k^B$ such that $\mathbb{E}^0[\phi_k(D)] = \mathbb{E}^0[\phi_k^B(D)]$ and $||\phi_k - \phi_k^B||_{L^2(\mathbb{P}^0)} < \epsilon$ for $k = 1, \dots, K$. Then,
    \allowdisplaybreaks
    \begin{align*}
       &\sqrt{m}\left(
   \begin{pmatrix}
    \frac{1}{n_1}\sum_{i=1}^{n_1}\phi_1(D_{1i})\\
    \vdots \\
     \frac{1}{n_K}\sum_{i=1}^{n_K}\phi_K(D_{Ki})
    \end{pmatrix} - 
    \begin{pmatrix}
    \mathbb{E}^0[\phi_1(D]\\
    \vdots \\
     \mathbb{E}^0[\phi_K(D)] 
    \end{pmatrix}\right) \\ &\quad\quad\quad\quad\quad\quad\quad\quad\quad =
        \sqrt{m}\left(
   \begin{pmatrix}
    \frac{1}{n_1}\sum_{i=1}^{n_1}\phi_1(D_{1i})\\
    \vdots \\
     \frac{1}{n_K}\sum_{i=1}^{n_K}\phi_K(D_{Ki})
    \end{pmatrix} - 
    \begin{pmatrix}
    \mathbb{E}^{1}[\phi_1(D]\\
    \vdots \\
     \mathbb{E}^{k}[\phi_K(D)] 
    \end{pmatrix}\right) \label{eq:a}\tag{a} \\& \quad\quad\quad\quad\quad\quad\quad\quad\quad
    + \sqrt{m}\left(\begin{pmatrix}
    \mathbb{E}^{1}[\phi_1(D)]\\
    \vdots \\
     \mathbb{E}^{k}[\phi_K(D)]
    \end{pmatrix} - 
    \begin{pmatrix}
    \mathbb{E}^{1}[\phi_1^B(D)]\\
    \vdots \\
     \mathbb{E}^{k}[\phi_K^B(D)] 
    \end{pmatrix}\right) \label{eq:b}\tag{b}
    \\& \quad\quad\quad\quad\quad\quad\quad\quad\quad
    + \sqrt{m}\left(\begin{pmatrix}
    \mathbb{E}^{1}[\phi_1^B(D)]\\
    \vdots \\
     \mathbb{E}^{k}[\phi_K^B(D)]
    \end{pmatrix} - 
    \begin{pmatrix}
    \mathbb{E}^0[\phi_1^B(D)]\\
    \vdots \\
     \mathbb{E}^0[\phi_K^B(D)] 
    \end{pmatrix}\right). \label{eq:c}\tag{c}
    \end{align*}
    Note that for any $\epsilon' > 0$ and for $k = 1, \dots, K$, 
    \begin{align*}
        &P\left(\sqrt{m} \Bigg| \frac{1}{n_k}\sum_{i=1}^{n_k}\phi_k(D_{ki}) - \mathbb{E}^{k}[\phi_k(D)] \Bigg| > \epsilon' \right) \\ & = E\left(P\left(\Bigg| \frac{1}{n_k}\sum_{i=1}^{n_k}\phi_k(D_{ki}) - \mathbb{E}^{k}[\phi_k(D)] \Bigg| > \frac{\epsilon'}{\sqrt{m}} \text{ } \Bigg | \text{  } W_k\right)\right) \\
        & \leq E\left( \frac{\mathbb{E}^{k}[\phi_k^2(D)]}{\epsilon'^2} \right) \cdot \frac{m}{n_k} = \frac{\mathbb{E}^0[\phi_k^2(D)]}{\epsilon'^2} \cdot \frac{m}{n_k} \xrightarrow[]{} 0,
    \end{align*}
    as the distributional uncertainty is of higher order than the sampling uncertainty. 
    Therefore, \eqref{eq:a} $ = o_p(1)$. Next let us investigate \eqref{eq:b}.  Recall that with $\psi_k = (\phi_k - \phi_k^B)\circ h$,
    \begin{align*}
       \sqrt{m}(\mathbb{E}^{k}[\phi_k(D)] - \mathbb{E}^{k}[\phi_k^B(D)]) &= \frac{\sqrt{m}\sum_{j=1}^{m}\int_{x \in I_j} \psi_k(x) dx \cdot (W_j^k - E[W^k])}{\sum_{j=1}^{m}W_j^k/m}.
    \end{align*}
    The variance of the numerator is bounded as
    \begin{align*}
        \text{Var}(W^k) \sum_{j=1}^{m}m\left(\int_{x \in I_j}\psi_k(x)dx\right)^2 
        & \leq \text{Var}(W^k) \sum_{j=1}^{m}\int_{x \in I_j}\psi_k^2(x)dx \\
        & = \text{Var}(W^k) \mathbb{E}^0[\psi_k^2(U)] \\
        & < \text{Var}(W^k) \cdot \epsilon^2
    \end{align*}
    where the first inequality holds by Jensen's inequality with $m \int_{x\in I_j} dx = 1$.  The denominator converges in probability to $E[W^k]$. Therefore, we can write
    \begin{equation*}
         \sqrt{m}(\mathbb{E}^{k}[\phi_k(D)] - \mathbb{E}^{k}[\phi_k^B(D)]) = \frac{\sqrt{m}\sum_{j=1}^{m}\int_{x \in I_j} \psi_k(x) dx \cdot (W_j^k - E[W^k])}{E[W^k]} + s_{n,k}
    \end{equation*}
    where $s_{n,k }$ is $o_p(1)$. Combining results, for $s_n = (s_{n,1}, \dots, s_{n, K})^{\intercal}$, we have that $E[ \eqref{eq:b} - s_{n}]=0$ and
    \begin{equation*}
        ||\text{Var}_P(\eqref{eq:b} - s_{n})||_{\infty} \leq \frac{||\text{Var}(\overrightarrow{W})||_{\infty}}{\min_k E^2[W^k]} \cdot \epsilon^2.
    \end{equation*}
    From Lemma \ref{lemma:pert_asymp}, we know that 
    \begin{equation*}
        \eqref{eq:c} \xrightarrow[]{d} N\left(0, \Sigma^W \odot \text{Var}_{\mathbb{P}^0}\left(\overrightarrow{\phi^B}(D)\right)\right).
    \end{equation*} 
     Let $||\phi_1^B||_2,\dots, ||\phi_K^B||_2, ||\phi_1||_2,\dots, ||\phi_K||_2 \leq M$. Note that for any $1 \leq k, k' \leq K$,
    \begin{align*}
        2|\mathbb{E}^0[\phi_k\phi_{k'}] - \mathbb{E}^0[\phi_k^B\phi_{k'}^B]| & \leq  |\mathbb{E}^0[(\phi_k - \phi_k^B)(\phi_{k'} + \phi_{k'}^B)]| + |\mathbb{E}^0[(\phi_k + \phi_k^B)(\phi_{k'} - \phi_{k'}^B)]| \\ & \leq 2M(||\phi_k - \phi_k^B||_2  + ||\phi_{k'} - \phi_{k'}^B||_2  ) \leq 4M\cdot \epsilon.
    \end{align*}
    Therefore,
    \begin{equation*}
        \Big|\Big|\text{Var}_{\mathbb{P}^0}(\overrightarrow{\phi}(D)) - \text{Var}_{\mathbb{P}^0}(\overrightarrow{\phi^B}(D))\Big|\Big|_{\infty} \leq 2M\cdot\epsilon.
    \end{equation*}
    Combining results, we have
    \begin{align*}
        \eqref{eq:a} +  \eqref{eq:b} +  \eqref{eq:c} &\stackrel{d}{=} 
        N\left(0, \Sigma^W \odot \text{Var}_{\mathbb{P}^0}(\overrightarrow{\phi}(D))\right) + o_p(1)  \\
        &+ (\eqref{eq:b} - s_n) + N\left(0, \Sigma^W \odot \left(\text{Var}_{\mathbb{P}^0}(\overrightarrow{\phi}(D)) - \text{Var}_{\mathbb{P}^0}(\overrightarrow{\phi^B}(D))\right)\right) \\
        & =  N\left(0, \Sigma^W \odot \text{Var}_{\mathbb{P}^0}(\overrightarrow{\phi}(D))\right) + O_p(\epsilon) + o_p(1). 
    \end{align*}
    Note that all the results hold for arbitrary $\epsilon > 0$. This completes the proof. 
\end{proof}

\subsection{Random target distribution}\label{sec:remark}

From Extended Theorem~\ref{extended-theorem} with $K+1$ datasets and $\Sigma^W \in \mathbb{R}^{(K+1)\times (K+1)}$, and considering that the sampling uncertainty is of low order than distributional uncertainty, for any Borel measurable square-integrable function $\phi: \mathcal{D} \xrightarrow[]{} \mathbb{R}^L$, we have
\begin{equation}
    \sqrt{m}\left(
   \begin{pmatrix}
    \frac{1}{n_1}\sum_{i=1}^{n_1}\phi(D_{1i})\\
    \vdots \\
     \frac{1}{n_K}\sum_{i=1}^{n_K}\phi(D_{Ki})
    \end{pmatrix} - 
    \begin{pmatrix}
    \mathbb{E}^{K+1}[\phi(D)]\\
    \vdots \\
     \mathbb{E}^{K+1}[\phi(D)] 
    \end{pmatrix}\right) \xrightarrow[]{d}
     N\left(0, (A\Sigma^W A^{\intercal}) \otimes  \text{Var}_{\mathbb{P}^0}({\phi}(D))\right),
\end{equation}
where  
\begin{equation*}
    A = \begin{pmatrix}
        1 & 0 & \dots & 0 & -1\\
         0 & 1 & \dots & 0 & -1\\
         \vdots & \vdots & & \vdots & \vdots\\
         0 & 0 & \dots & 1 & -1
    \end{pmatrix}  \in \mathbb{R}^{K\times (K+1)}.
\end{equation*}
Therefore, with $\mathbb{P}^{K+1}$ as a new target distribution, we still have Extended Theorem~\ref{thm:perturbation-model} but with a different distributional covariance matrix $A\Sigma^WA^{\intercal}$.

\subsection{Proof of Remark~\ref{remark:var}}\label{sec:remark-2}

Note that $\widehat{\text{Var}}_{\mathbb{P}^0}(\phi(D))$ is an empirical variance of $\phi$ on the pooled donor data $\{(D_{ki})_{i=1}^{n_k}\}_{k=1}^K$ and can be written as 
\begin{align*}
    \widehat{\text{Var}}_{\mathbb{P}^0}(\phi(D)) 
    & =   \sum_{k=1}^{K}\frac{n_k}{n}\frac{1}{n_k}\sum_{i=1}^{n_k}\phi(D_{ki})\phi(D_{ki})^{\intercal} \\ &- \left(\sum_{k=1}^{K}\frac{n_k}{n}\frac{1}{n_k}\sum_{i=1}^{n_k}\phi(D_{ki})\right)\left(\sum_{k=1}^{K}\frac{n_k}{n}\frac{1}{n_k}\sum_{i=1}^{n_k}\phi(D_{ki})\right)^{\intercal}.
\end{align*} 
where $n = n_1 + \dots + n_K$. 
As $\phi$ has a finite fourth moment,
by Theorem~\ref{thm:perturbation-model}, we have that 
\begin{align*}
    \widehat{\text{Var}}_{\mathbb{P}^0}(\phi(D))  &= \sum_{k=1}^K \frac{n_k}{n}\left( \mathbb{E}^0[\phi(D)\phi(D)^{\intercal}] + o_p(1) \right) \\ &- \left(\sum_{k=1}^K \frac{n_k}{n}\left( \mathbb{E}^0[\phi(D)] + o_p(1)\right)\right)\left(\sum_{k=1}^K \frac{n_k}{n}\left( \mathbb{E}^0[\phi(D)] + o_p(1)\right)\right)^{\intercal} \\
    & = \text{Var}_{\mathbb{P}^0}(\phi(D)) + o_p(1).
\end{align*}

\section{Special cases of the random perturbation model.}\label{sec:special-cases}
We present special cases of the random perturbation model to help develop some intuition about potential applications.

\begin{example}[Random shift across time]
Consider the case where $W_j^k = W_j^{k-1} + \epsilon_{k}$, where the $\epsilon_{k}$ have mean zero and are uncorrelated for $k=1,\ldots,K$. This can be used to model random distributional shifts across time. For example, at some time point, we might see (randomly) higher or lower percentage of old patients at a hospital. Similarly, one could also have an auto-regressive model $W_j^k = \frac{1}{2} W_j^{k-1} + \frac{1}{2} W_j^{k-2} + \epsilon_{k}$. Please note that these auto-regressive models are defined on the weights of the distribution and not the data itself. This allows us to model random distributional shift for different data modalities such as discrete, ordinal, or continuous sample spaces, with a single framework.
\end{example}

\begin{example}[Independent shifts and meta analysis]
  If the distributional perturbations $W_j^k$ are i.i.d.\ across $k=1,\ldots,K$, then $\Sigma^W$ is a diagonal matrix. Thus, for any $\phi \in L^2(\mathbb{P}^0)$, $\hat \theta_k := \frac{1}{n_k}\sum_{i=1}^{n_k}\phi(D_{ki})$ are asymptotically Gaussian and uncorrelated across $k=1,\ldots,K$. This implies a meta-analytic model
  \begin{equation}
      \hat \theta_k \stackrel{\text{i.i.d.}}{\approx} N(\theta^0,\tau^2),
  \end{equation}
where $\tau^2 = \frac{\Sigma_{kk}^W}{m} \text{Var}_{\mathbb{P}^0}(\phi)$ and $\theta^0 := \mathbb{E}^0[\phi(D)]$. Thus, in this special case, the random perturbation model justifies a random effect meta-analysis. 
\end{example}

\begin{example}[Random shift between locations]
In general, the meta-analytic model might not hold since the distributional perturbations might not be independent or identically distributed across $k=1,\ldots,K$. As an example, data from California might be similar to data from Washington, but very different from data from Pennsylvania. In the synthetic control literature, it is popular to model data distributions from different locations as (approximate) mixtures of each other. For example, data from California might be well approximated by a mixture of Washington and Nevada. In our framework, this can be expressed as
\begin{equation}
    W^{\text{California}} = .7 \cdot W^{\text{Washington}} + .3 \cdot W^{\text{Nevada}} + \epsilon,
\end{equation}
where $\epsilon$ is mean-zero noise.
\end{example}

\section{Proof of Lemma~\ref{lemma:ooderror}}\label{sec:ood-error-proof}

\begin{proof}
The first part of proof follows \cite{van2000asymptotic}, Theorem 5.41 with 
\begin{align*}
    \Psi_n(\theta) &= \sum_{k=1}^K \beta_k \hat{\mathbb{E}}^k[\partial_\theta \mathcal{L}(\theta, D)] = \sum_{k=1}^K \beta_k \frac{1}{n_k} \sum_{i=1}^{n_k} \partial_\theta \mathcal{L}(\theta,D_{ki}), \\
    \dot{\Psi}_n(\theta) &= \sum_{k=1}^K \beta_k \frac{1}{n_k} \sum_{i=1}^{n_k} \partial_\theta^2 \mathcal{L}(\theta,D_{ki}).
\end{align*}
By Taylor's theorem there exists a (random vector) $\tilde{\theta}$ on the line segment between $\theta^0$ and $\hat{\theta}^{\beta}$ such that 
\begin{equation*}
    0 = \Psi_n(\hat{\theta}) = \Psi_n(\theta^0) +  \dot{\Psi}_n(\theta^0) (\hat{\theta}^{\beta} - \theta^0)+ \frac{1}{2}(\hat{\theta}^{\beta} - \theta^0)^{\intercal} \Ddot{\Psi}_n(\tilde \theta) (\hat{\theta}^{\beta} - \theta^0).
\end{equation*}
By Extended Theorem~\ref{extended-theorem} with $\phi(D) = \partial_\theta^2 \mathcal{L}(\theta^0,D)$, we have
\begin{align}
    \dot{\Psi}_n(\theta^0)  &= \sum_{k=1}^K\beta_k\frac{1}{n_k} \sum_{i=1}^{n_k} \partial_\theta^2 \mathcal{L}(\theta^0,D_{ki}) \nonumber \\ &= \sum_{k=1}^K \beta_k \mathbb{E}^0[ \partial_\theta^2 \mathcal{L}(\theta^0,D) ] + o_P(1) = \mathbb{E}^0[ \partial_\theta^2 \mathcal{L}(\theta^0,D) ] + o_P(1). \label{eq:lemma1-pf-1}
\end{align}
By assumption, there exists a ball $B$ around $\theta^0$ such that $\theta \xrightarrow[]{} \partial_\theta^3 \mathcal{L}(\theta,D)$ is dominated by a fixed function $h(\cdot)$ for every $\theta \in B$. The probability of the event $\{\hat{\theta}^{\beta} \in B\}$ tends to 1. On this event,
\begin{equation*}
    \|\Ddot{\Psi}_n(\tilde \theta)\| = \Big|\Big| \sum_{k=1}^K\beta_k\frac{1}{n_k} \sum_{i=1}^{n_k} \partial_\theta^3 \mathcal{L}(\tilde \theta,D_{ki})\Big|\Big| \leq \sum_{k=1}^K|\beta_k|\frac{1}{n_k} \sum_{i=1}^{n_k} h(D_{ki}).
\end{equation*}
Using Extended Theorem~\ref{extended-theorem} with $\phi(D) = h(D)$, we get
\begin{equation}\label{eq:lemma1-pf-2}
       \| \Ddot{\Psi}_n(\tilde \theta) \|   \le \sum_{k=1}^K|\beta_k|\frac{1}{n_k} \sum_{i=1}^{n_k} h(D_{ki}) = O_P(1).
\end{equation}
Combining \eqref{eq:lemma1-pf-1} and \eqref{eq:lemma1-pf-2} gives us
\begin{align*}
    -\Psi_n(\theta^0) & = \left(\mathbb{E}^0[ \partial_\theta^2 \mathcal{L}(\theta^0,D) ] + o_P(1) + \frac{1}{2}(\hat{\theta}^{\beta} - \theta^0)\text{ }O_P(1)\right)(\hat{\theta}^{\beta} - \theta^0) \\ &= \left(\mathbb{E}^0[ \partial_\theta^2 \mathcal{L}(\theta^0,D) ] + o_P(1) \right)(\hat{\theta}^{\beta} - \theta^0).
\end{align*}
The probability that the matrix $\mathbb{E}^0[ \partial_\theta^2 \mathcal{L}(\theta^0,D) ] + o_P(1) $ is invertible tends to 1.  Multiplying the preceding equation by $\sqrt{m}$ and applying $(\mathbb{E}^0[ \partial_\theta^2 \mathcal{L}(\theta^0,D) ] + o_P(1))^{-1} $ left and right gives us that
\begin{equation*}
    \sqrt{m}\left(\hat{\theta}^{\beta} - \theta^0\right) = \sqrt{m}\left(\sum_{k=1}^K \beta_k \hat{\mathbb{E}}^k[\phi(D)] - \mathbb{E}^0[\phi(D)]\right) + o_P(1),
\end{equation*}
where $\phi(D) = - \mathbb{E}^0[\partial_\theta^2 \mathcal{L}(\theta^0,D)]^{-1}  \partial_\theta \mathcal{L}(\theta^0,D)$. By Extended Theorem~\ref{extended-theorem}, we have
\begin{equation}\label{eq:theta-asymp}
    \sqrt{m}\left(\hat{\theta}^{\beta} - \theta^0\right) \xrightarrow[]{d} N(0, \beta^{\intercal}\Sigma^W \beta \text{ }\cdot \text{Var}_{\mathbb{P}^0}(\phi(D)))
\end{equation}
Now let $\Phi(\theta) =  \mathbb{E}^0[\partial_{\theta}\mathcal{L}(\theta, D)], \dot{\Phi}(\theta) = \mathbb{E}^0[\partial_{\theta}^2\mathcal{L}(\theta, D)], \Ddot{\Phi}(\theta) = \mathbb{E}^0[\partial_{\theta}^3\mathcal{L}(\theta, D)]$. By Taylor's theorem, there exist $\tilde{\theta}, \tilde{\tilde{\theta}}$,  on the line segment between $\theta^0$ and $\hat{\theta}^{\beta}$ such that 
\begin{align*}
     \mathbb{E}^0[\mathcal{L}(\hat \theta^\beta,D)]  -  \mathbb{E}^0[\mathcal{L}( \theta^0,D)] &= (\hat{\theta}^{\beta} - \theta^0)^{\intercal}{\Phi}(\theta^0)  + \frac{1}{2}(\hat{\theta}^{\beta} - \theta^0)^{\intercal} \dot{\Phi}(\tilde \theta) (\hat{\theta}^{\beta} - \theta^0) \\
     &= \frac{1}{2}(\hat{\theta}^{\beta} - \theta^0)^{\intercal} \left( \dot{\Phi}(\theta^0) + \Ddot{\Phi}(\tilde{\tilde{\theta}}) (\tilde{\theta} - \theta^0)\right) (\hat{\theta}^{\beta} - \theta^0) \\
     & = \frac{1}{2}(\hat{\theta}^{\beta} - \theta^0)^{\intercal} \left( \dot{\Phi}(\theta^0) + o_P(1) \right) (\hat{\theta}^{\beta} - \theta^0),
\end{align*}
as ${\Phi}(\theta^0) = 0$ and $\Ddot{\Phi}(\tilde{\tilde{\theta}}) \leq \mathbb{E}^0[h(D)]$ with probability approaching 1. Then  the rescaled excess risk can be written as 
\begin{equation}\label{eq:risk-taylor}
  m \cdot \left( \mathbb{E}^0[\mathcal{L}(\hat \theta^\beta,D)]  -  \mathbb{E}^0[\mathcal{L}( \theta^0,D)] \right) = m \cdot (\hat \theta^\beta - \theta^0)^\intercal \mathbb{E}^0[\partial_\theta^2 \mathcal{L}(\theta^0,D)]   (\hat \theta^\beta - \theta^0) + o_P(1).
\end{equation}
Using \eqref{eq:theta-asymp}, asymptotically \eqref{eq:risk-taylor} follows a distribution with asymptotic mean
\begin{equation*}
 \mu^\beta =  \beta^\intercal \Sigma^W \beta \cdot \mathrm{Trace}( \mathbb{E}^0[\partial_\theta^2 \mathcal{L}(\theta^0,D)]^{-1} \mathrm{Var}_{\mathbb{P}^0}(  \partial_\theta \mathcal{L}(\theta^0,D) ))
\end{equation*}
This completes the proof. 

In the following, we add regularity conditions that lead to the consistency of M-estimators, $\hat{\theta}^{\beta}$. 

\begin{lemma}[Consistency of M-estimators]\label{lemma:consistency} Consider the M-estimator
$$
\hat{\theta}^{\beta} = \arg \min_{\theta \in \Omega} \sum_{k=1}^K \beta_k \hat{\mathbb{E}}^k[\mathcal{L}(\theta, D)] 
$$
and the target $\theta^0 = \arg \min_{\theta\in \Omega}\mathbb{E}^0[\mathcal{L}(\theta, D)]$, where $\Omega$ is a compact subset of $\mathbb{R}^d$. 
    Assume that $\theta \xrightarrow[]{} \mathcal{L}(\theta, D)$ is continuous and that $\inf_{||\theta - \theta'||_2 \leq \delta}\mathcal{L}(\theta, D)$ is square-integrable under $\mathbb{P}^0$ for every $\delta$ and $\theta'$ and that $\inf_{\theta \in \Omega}\mathcal{L}(\theta, D)$ is square-integrable. We assume that $\mathbb{E}^0[\mathcal{L}(\theta, D)]$ has a unique minimum. Then, $\hat{\theta}^{\beta} - \theta^0 = o_P(1)$.
\end{lemma}

\begin{proof}
    The proof follows \cite{van2000asymptotic}, Theorem 5.14 with $m_\theta(D) = - \mathcal{L}(\theta,D)$. 
    
    Fix some $\theta$ and let $U_{\ell} \downarrow \theta$ be a decreasing sequence of open balls around $\theta$ of diameter converging to zero. Write $m_{U}(D)$ for $\sup_{\theta \in U}m_{\theta}(D)$. The sequence $m_{U_{\ell}}$ is decreasing and greater than $m_{\theta}$ for every $\ell$. As $\theta \xrightarrow[]{} m_{\theta}(D)$ is continuous, by monotone convergence theorem, we have $\mathbb{E}^0[m_{U_{\ell}}(D)] \downarrow \mathbb{E}^0[m_{\theta}(D)]$. 

    For $\theta \neq \theta^0$, we have $\mathbb{E}^0[m_\theta(D)] < \mathbb{E}^0[m_{\theta^0}(D)]$.  Combine this with the preceding paragraph to see that for every $\theta \neq \theta^0$ there exits an open ball $U_{\theta}$ around $\theta$ with $\mathbb{E}^0[m_{U_{\theta}}(D)] < \mathbb{E}^0[m_{\theta^0}(D)]$. For any given $\epsilon > 0$, let the set $B = \{\theta \in \Omega : ||\theta - \theta^0||_2 \geq \epsilon\}$. The set $B$ is compact and is covered by the balls $\{U_{\theta}: \theta \in B\}$. Let $U_{\theta_1},\dots, U_{\theta_p}$ be a finite sub-cover. Then, 
\begin{align}
    \sup_{\theta \in B} \sum_{k=1}^K\beta_k \frac{1}{n_k} \sum_{i=1}^{n_k} m_{\theta}(D_{ki}) & \leq \sup_{j = 1, \dots, p} \sum_{k=1}^K\beta_k \frac{1}{n_k} \sum_{i=1}^{n_k} m_{U_{\theta_j}}(D_{ki})
    \nonumber \\ &= \sup_{j = 1, \dots, p} \mathbb{E}^0[ m_{U_{\theta_j}}(D)] + o_P(1) < \mathbb{E}^0[m_{\theta^0}(D)] + o_P(1) \label{eq:pf-lemma2},
\end{align}
where for the equality we apply Theorem~\ref{thm:perturbation-model} with $\phi(D) = m_{U_{\theta_j}}(D)$ for all $j=1,\ldots,p$. 

If $\hat{\theta}^{\beta} \in B$, then 
\begin{equation*}
    \sup_{\theta \in B} \sum_{k=1}^K\beta_k \frac{1}{n_k} \sum_{i=1}^{n_k} m_{\theta}(D_{ki}) \geq \sum_{k=1}^K\beta_k \frac{1}{n_k} \sum_{i=1}^{n_k} m_{\hat{\theta}^{\beta}}(D_{ki})  \geq \sum_{k=1}^K\beta_k \frac{1}{n_k} \sum_{i=1}^{n_k} m_{\theta^0}(D_{ki}) - o_P(1),
\end{equation*}
where the last inequality comes from the definition of $\hat{\theta}^{\beta}$.
Using Theorem~\ref{thm:perturbation-model} with $\phi(D) = m_{\theta^0}(D)$, we have
\begin{equation*}
   \sum_{k=1}^K\beta_k \frac{1}{n_k} \sum_{i=1}^{n_k} m_{\theta^0}(D_{ki}) - o_P(1) = \mathbb{E}^0[m_{\theta^0}(D)] - o_P(1).
\end{equation*}
Therefore, 
\begin{equation*}
    \{\hat{\theta}^{\beta} \in B\} \subset \Bigg\{ \sup_{\theta \in B} \sum_{k=1}^K\beta_k \frac{1}{n_k} \sum_{i=1}^{n_k} m_{\theta}(D_{ki}) \geq \mathbb{E}^0[m_{\theta^0}(D)] - o_P(1) \Bigg\}.
\end{equation*}
By the equation \eqref{eq:pf-lemma2}, the probability of the event on the right hand side converges to zero as $n \xrightarrow[]{} \infty$. This completes the proof. 
\end{proof}

\end{proof}

\section{Proof of Theorem \ref{theorem:dlm}}\label{sec:dlmproofs}

Define the feature matrix before reparametrization as 
\begin{equation*}
	\Phi = \begin{pmatrix}
		\hat{\mathbb{E}}^1[\phi_{1}] - \hat{\mathbb{E}}^0[\phi_1]& \dots & \hat{\mathbb{E}}^K[\phi_{1}]  - \hat{\mathbb{E}}^0[\phi_1]\\
		\vdots & & \vdots \\
		\hat{\mathbb{E}}^1[\phi_{L}]  - \hat{\mathbb{E}}^0[\phi_L]& \dots & \hat{\mathbb{E}}^K[\phi_{L}] - \hat{\mathbb{E}}^0[\phi_L]
	\end{pmatrix} \in \mathbb{R}^{L \times K},
\end{equation*}
where the $\ell$-th row is defined as $\Phi_{\ell}$.

 By Theorem~\ref{thm:perturbation-model} and using that the sampling uncertainty is of low order than distributional uncertainty, we have
    \begin{equation*}
        \sqrt{m} \Phi_{\ell} := \sqrt{m}\begin{pmatrix}
            \hat{\mathbb{E}}^1[\phi_{\ell}] - \hat{\mathbb{E}}^0[\phi_{\ell}] \\
            \vdots \\
            \hat{\mathbb{E}}^K[\phi_{\ell}] - \hat{\mathbb{E}}^0[\phi_{\ell}] 
        \end{pmatrix}
        \stackrel{d}{=} \bm{Z}_{\ell} + o_p(1)
    \end{equation*}
    where $\bm{Z}_1, \dots, \bm{Z}_L \in \mathbb{R}^K$ are i.i.d Gaussian random variables following $N(0, \Sigma^W)$. Let $\bm{Z} \in  \mathbb{R}^{L \times K}$ be a matrix in which the $\ell$-th row is $\bm{Z}_{\ell}$.

       Note that
    \begin{equation*}
    \hat{\beta} = \argmin_{\beta:\beta^{\intercal}1 = 1} \beta^{\intercal}\left(\sum_{\ell=1}^L \Phi_{\ell}\Phi_{\ell}^{\intercal}\right)\beta = \arg \min_{\beta:\beta^{\intercal}1 = 1}\beta^{\intercal} \bm{\Phi}^{\intercal}\bm{\Phi} \beta.
    \end{equation*}
    The closed form of $\hat{\beta}$ is known as
    \begin{equation*}
        \hat{\beta} = \frac{(\bm{\Phi}^{\intercal}\bm{\Phi})^{-1}\bm{1}}{\bm{1}^{\intercal}(\bm{\Phi}^{\intercal}\bm{\Phi})^{-1}\bm{1}} \stackrel{d}{=} \frac{(\bm{Z}^{\intercal}\bm{Z})^{-1}\bm{1}}{\bm{1}^{\intercal}(\bm{Z}^{\intercal}\bm{Z})^{-1}\bm{1}} + o_p(1).
    \end{equation*}
    We write $\hat{\beta}_Z = \frac{(\bm{Z}^{\intercal}\bm{Z})^{-1}\bm{1}}{\bm{1}^{\intercal}(\bm{Z}^{\intercal}\bm{Z})^{-1}\bm{1}} $.
    Moreover, the closed form of $\beta^*$ is
    \begin{equation*}
        \beta^* = \argmin_{\beta:\beta^{\intercal}1 = 1} \beta^{\intercal}\Sigma^W\beta = \frac{(\Sigma^W)^{-1}\bm{1}}{\bm{1}^{\intercal}(\Sigma^W)^{-1}\bm{1}}.
    \end{equation*}
    Let $Q$ be a $p \times (K-1)$ non-random matrix with $\text{rank}(Q) =p \leq K-1$.
    Let
    $\Tilde{Q} = (Q'^{\intercal}, \bm{1})^{\intercal} \in \mathbb{R}^{(p+1)\times K}$ where 
    $$Q' = Q \begin{pmatrix}
        \bm{I}_{(K-1) \times (K-1)} & \bm{0}_{(K-1) \times 1} 
    \end{pmatrix} \in \mathbb{R}^{p \times K}.
    $$ 
    Note that the $\text{rank}(\Tilde{Q}) = p + 1 \leq K$. Define $\Tilde{S} := \tilde{Q}(\bm{Z}^{\intercal}\bm{Z})^{-1}\Tilde{Q}^{\intercal}$ where
    \begin{equation*}
        \Tilde{S} = \begin{pmatrix}
             \Tilde{S}_{11} & \Tilde{S}_{12} \\
             \Tilde{S}_{21} & \Tilde{S}_{22}
        \end{pmatrix} = \begin{pmatrix}
            Q'(\bm{Z}^{\intercal}\bm{Z})^{-1}Q'^{\intercal} &  Q'(\bm{Z}^{\intercal}\bm{Z})^{-1}\bm{1}\\
             \bm{1}^{\intercal}(\bm{Z}^{\intercal}\bm{Z})^{-1}Q'^{\intercal} & \bm{1}^{\intercal}(\bm{Z}^{\intercal}\bm{Z})^{-1}\bm{1}.
        \end{pmatrix}
    \end{equation*}
     Similarly, let $\Tilde{P} := \tilde{Q}(\Sigma^W)^{-1}\Tilde{Q}^{\intercal}$ where $\Tilde{P}_{11} =  Q'(\Sigma^W)^{-1}Q'^{\intercal}, \Tilde{P}_{12} = Q'(\Sigma^W)^{-1}\bm{1},$ and $\Tilde{P}_{22} =  \bm{1}^{\intercal}(\Sigma^W)^{-1}\bm{1}$. 

     \paragraph{Results from Multivariate Statistical Theory.} Note that $\bm{Z}^{\intercal}\bm{Z} \sim \text{Wish}_K(L; \Sigma^W)$ which is the Wishart distribution with a degree of freedom $L$ and a scale matrix $\Sigma^W$. By Theorem 3.2.11 and 3.6 in \cite{muirhead1982aspects}, we have
    \begin{equation*}
        \Tilde{S} \sim W_{p+1}^{-1}(L-K+2p + 3; \Tilde{P}).
    \end{equation*}
    which is the inverse-Wishart distribution with a degree of freedom $L - K + 2p + 3$ and a scale matrix $\Tilde{P}$.
     By Theorem 3 (b) and (d) of \cite{Bodnar2008Wishart}, we have that
    \begin{equation*}
        \Tilde{S}_{12}\Tilde{S}_{22}^{-1} | \Tilde{S}_{11.2} \sim N(\Tilde{P}_{12}\Tilde{P}_{22}^{-1}, \Tilde{S}_{11.2} \otimes \Tilde{P}_{22}^{-1}).
    \end{equation*}
    Therefore, we get
    \begin{equation*}
        \Tilde{P}_{22}^{1/2} \Tilde{S}_{11.2}^{-1/2} \left(
\Tilde{S}_{12}\Tilde{S}_{22}^{-1} - \Tilde{P}_{12}\Tilde{P}_{22}^{-1}
        \right)| \Tilde{S}_{11.2} \sim N(0, \bm{I}).
    \end{equation*}
    The right hand side does not depend on $\Tilde{S}_{11.2}$ anymore. Therefore, we have
    \begin{equation}
        \Tilde{P}_{22}^{1/2} \Tilde{S}_{11.2}^{-1/2}  \left( Q(\hat{\beta}_Z)_{1:(K-1)}  -  Q{\beta}^*_{1:(K-1)}  \right)\sim N(\bm{0}, \bm{I}).
    \end{equation}
    By Theorem 3.2.12 in \cite{muirhead1982aspects}, we have
    \begin{equation*}
         \frac{\Tilde{P}_{22}}{\Tilde{S}_{22}} = \frac{\bm{1}^{\intercal}(\Sigma^W)^{-1}\bm{1}}{\bm{1}^{\intercal}(Z^{\intercal}Z)^{-1}\bm{1}} \sim \chi^2(L-K+1).
    \end{equation*}
    By Theorem 3 (e) of \cite{Bodnar2008Wishart}, $\Tilde{S}_{22}$ is independent of $\Tilde{S}_{12}\Tilde{S}_{22}^{-1}$ and $\Tilde{S}_{11.2}$. Hence, we get
    \begin{equation*}
    \sqrt{L - K + 1} \cdot \Tilde{S}_{22}^{1/2}\Tilde{S}_{11.2}^{-1/2}  \left( Q(\hat{\beta}_Z)_{1:(K-1)}  -  Q{\beta}^*_{1:(K-1)}  \right)\sim t_{p}(L-K+1; \bm{0}, \bm{I}),
    \end{equation*}
    where $ t_{p}(L-K+1; \bm{0}, \bm{I})$ is $p$-dimensional (centered) $t$-distribution with $L-K+1$ degrees of freedom.
    Now using $\Phi = \bm{Z}/\sqrt{m} + o_p(1/\sqrt{m})$ from Theorem~\ref{thm:perturbation-model}, we get
    \begin{align}
        \sqrt{L - K + 1}\sqrt{\bm{1}^{\intercal}(\bm{\Phi}^{\intercal}\bm{\Phi})^{-1}\bm{1}} \cdot \left[QRQ^{\intercal}\right]^{-1/2} \left( Q\hat{\beta}_{1:(K-1)}  -  Q{\beta}^*_{1:(K-1)}  \right)\\ \stackrel{d}{=} t_{p}(L-K+1; \bm{0}, \bm{I}) + o_p(1)\label{eq:tdist}.
    \end{align}
    where
    \begin{equation*}
        R = \begin{pmatrix}
        \bm{I}_{(K-1) \times (K-1)} & \bm{0}_{(K-1) \times 1} 
    \end{pmatrix}  \left((\bm{\Phi}^{\intercal}\bm{\Phi})^{-1} - \frac{(\bm{\Phi}^{\intercal}\bm{\Phi})^{-1}\bm{1}\bm{1}^{\intercal}(\bm{\Phi}^{\intercal}\bm{\Phi})^{-1}}{\bm{1}^{\intercal}(\bm{\Phi}^{\intercal}\bm{\Phi})^{-1}\bm{1}} \right) \begin{pmatrix}
        \bm{I}_{(K-1) \times (K-1)} & \bm{0}_{(K-1) \times 1} 
    \end{pmatrix}^{\intercal}.
    \end{equation*}

    \paragraph{Connection with Reparametrized Linear Regression.}
    In this part, we will show that 
    \begin{equation}\label{eq:varx}  (\Tilde{\bm{\Phi}}^{\intercal}\Tilde{\bm{\Phi}})^{-1} = R.
    \end{equation}
    Using the closed form of $\hat{\beta}$, we have that
    \begin{equation}\label{eq:rss}
         (L - K + 1) \cdot \hat{\sigma}^2 = \hat{\beta}^{\intercal}\Phi^{\intercal}\Phi \hat{\beta} = \frac{1}{\bm{1}^{\intercal}(\Phi^{\intercal}\Phi)^{-1}\bm{1}}.
    \end{equation}
Then combining \eqref{eq:tdist}, \eqref{eq:varx}, and \eqref{eq:rss}, we have 
 \begin{equation*}
         \left(Q (\Tilde{\Phi}^{\intercal}\Tilde{\Phi})^{-1}\hat{\sigma}^2Q^{\intercal}\right)^{-\frac{1}{2}} \left(Q \hat{\beta}_{1:(K-1)} - Q \beta^*_{1:(K-1)}\right) \stackrel{d}{=} t_p(L - K + 1; \bm{0}, \bm{I}) +o_p(1).
     \end{equation*}
     Note that this is a more general version of Theorem~\ref{theorem:dlm}. By letting $Q$ be an $(K-1)\times(K-1)$ identity matrix, we have Theorem~\ref{theorem:dlm}.  For the simplicity of notation, let $V =\bm{\Phi}^{\intercal}\bm{\Phi}$. Note that 
     \begin{align*}
        \Tilde{\bm{\Phi}}^{\intercal}\Tilde{\bm{\Phi}} &= 
        \begin{pmatrix}
            \bm{I}_{(K-1)\times(K-1)} & -\bm{1}
        \end{pmatrix} 
        \begin{pmatrix}
        V_{11} & V_{12} \\V_{12}^{\intercal} & V_{22}
        \end{pmatrix}
        \begin{pmatrix}
            \bm{I}_{(K-1)\times(K-1)} & -\bm{1}
        \end{pmatrix}^{\intercal} \\
        & = V_{11} - \bm{1}V_{12}^{\intercal} - V_{12}\bm{1}^{\intercal} + \bm{1}V_{22}\bm{1}^{\intercal}
     \end{align*}
     where $V_{22} \in \mathbb{R}$ and $\bm{1} \in \mathbb{R}^{(K-1)\times 1}$. Then, by applying Sherman-Morrison lemma, we have
     \begin{align*}
         (\Tilde{\bm{\Phi}}^{\intercal}\Tilde{\bm{\Phi}})^{-1} & = \left(V_{11.2} + V_{12}V_{22}^{-1}V_{12}^{\intercal}-\bm{1}V_{12}^{\intercal} - V_{12}\bm{1}^{\intercal} + \bm{1}V_{22}\bm{1}^{\intercal}\right)^{-1}\\
         & = \left(V_{11.2} + \left(V_{12}V_{22}^{-1/2} - V_{22}^{1/2}\bm{1}\right)\left(V_{12}V_{22}^{-1/2} - V_{22}^{1/2}\bm{1}\right)^{\intercal}\right)^{-1} \\
         & = V_{11.2}^{-1} - \frac{V_{11.2}^{-1}\left(V_{12}V_{22}^{-1/2} - V_{22}^{1/2}\bm{1}\right)\left(V_{12}V_{22}^{-1/2} - V_{22}^{1/2}\bm{1}\right)^{\intercal}V_{11.2}^{-1}}{1 + \left(V_{12}V_{22}^{-1/2} - V_{22}^{1/2}\bm{1}\right)^{\intercal}V_{11.2}^{-1}\left(V_{12}V_{22}^{-1/2} - V_{22}^{1/2}\bm{1}\right)}.
     \end{align*}
     By using the inverse of block matrix as
     \begin{equation*}
         V^{-1} = \begin{pmatrix}
             V_{11.2}^{-1} & - V_{11.2}^{-1}V_{12}V_{22}^{-1} \\
             - V_{22}^{-1}V_{12}^{\intercal}V_{11.2}^{-1} & V_{22}^{-1} + V_{22}^{-1}V_{12}^{\intercal}V_{11.2}^{-1}V_{12}V_{22}^{-1}
         \end{pmatrix},
     \end{equation*} we also get
     \begin{align*}
         R = V_{11.2}^{-1} -\frac{V_{11.2}^{-1}\left(V_{12}V_{22}^{-1} - \bm{1}\right)\left(V_{12}V_{22}^{-1} - \bm{1}\right)^{\intercal}V_{11.2}^{-1}}{
         \bm{1}^{\intercal}V_{11.2}^{-1}\bm{1} - V_{22}^{-1}V_{12}^{\intercal}V_{11.2}^{-1}\bm{1} - \bm{1}^{\intercal}V_{11.2}^{-1}V_{12}V_{22}^{-1} + V_{22}^{-1} + V_{22}^{-1}V_{12}^{\intercal}V_{11.2}^{-1}V_{12}V_{22}^{-1}
         }.
     \end{align*}
     Then, we can easily check that $(\Tilde{\bm{\Phi}}^{\intercal}\Tilde{\bm{\Phi}})^{-1} = R$. This completes the proof.

\section{Confidence Intervals for Target Parameters}\label{sec:prediction-intervals}

First, we will discuss how to form asymptotically valid confidence intervals for $\mathbb{E}^0[\phi_0(D)]$. Assume that test functions $\phi_1(X), \dots, \phi_L(X)$ are uncorrelated and have unit variances under $\mathbb{P}^0$. 

Let's say $\hat{\beta}$ in \eqref{eq:least-squares} is consistent as follows,
\begin{equation*}
        \hat{\beta} \xrightarrow[]{p} \beta^* := \argmin_{\beta:\beta^{\intercal}1} \beta^{\intercal}\Sigma^W\beta.
\end{equation*}
Then, we have for $\overrightarrow{\hat{\mathbb{E}}[\phi_0]} = (\hat{\mathbb{E}}^1[\phi_0], \dots, \hat{\mathbb{E}}^K[\phi_0])^{\intercal}$,
\begin{align*}
\sqrt{m}\left(\sum_{k=1}^{K} \hat{\beta}_k \hat{\mathbb{E}}^k[\phi_0] - \mathbb{E}^0[\phi_0]\right) &= \sqrt{m}\left(\hat{\beta} - \beta^*\right)^{\intercal}\left(\overrightarrow{\hat{\mathbb{E}}[\phi_0]} - \mathbb{1} \cdot \mathbb{E}^0[\phi_0]\right) \\
&+ \sqrt{m} \text{  } {\beta^*}^{\intercal} \left(\overrightarrow{\hat{\mathbb{E}}[\phi_0]} - \mathbb{1} \cdot \mathbb{E}^0[\phi_0]\right)
\end{align*}
since $\hat{\beta}^{\intercal}\mathbb{1} = 1$ and $\beta^*{}^{\intercal}\mathbb{1} = 1$. By Theorem~\ref{thm:perturbation-model}, 
\begin{equation*}
    \sqrt{m} \left(\overrightarrow{\hat{\mathbb{E}}[\phi_0]} - \mathbb{1} \cdot \mathbb{E}^0[\phi_0]\right) \xrightarrow[]{d} N(0, \text{Var}_{\mathbb{P}^0}(\phi_0) \cdot \Sigma^W). 
\end{equation*}
Therefore, 
\begin{align}
\sqrt{m}\left(\sum_{k=1}^{K} \hat{\beta}_k \hat{\mathbb{E}}^k[\phi_0] - \mathbb{E}^0[\phi_0]\right) &=  \sqrt{m} \text{  } {\beta^*}^{\intercal} \left(\overrightarrow{\hat{\mathbb{E}}[\phi_0]} - \mathbb{1} \cdot \mathbb{E}^0[\phi_0]\right)  + o_p(1) \nonumber \\
&\xrightarrow[]{d} N(0, \text{Var}_{\mathbb{P}^0}(\phi_0) \cdot  {\beta^*}^{\intercal} \Sigma^W \beta^*). \label{eq:ci-part1}
\end{align}
Then, we can construct confidence interval for $\mathbb{E}^0[\phi_0]$ as
\begin{equation*}
    \sum_{k=1}^{K} \hat{\beta}_k \hat{\mathbb{E}}^k[\phi_0] \pm z_{1-\alpha/2} \cdot  \sqrt{\text{Var}_{\mathbb{P}^0}(\phi_0)} \cdot \sqrt{\frac{{\beta^*}^{\intercal} \Sigma^W \beta^*}{m}}.
\end{equation*}
How can we obtain a consistent estimator for ${\beta^*}^{\intercal} \Sigma^W \beta^* / m$? By Theorem~\ref{thm:perturbation-model},
\begin{equation*}
   \hat{\mathbb{E}}^0[\phi_{\ell}] - \sum_{k=1}^{K} \beta_k^* \hat{\mathbb{E}}^k[\phi_{\ell}] \text{  } \stackrel{d}{=} \text{  }\sqrt{\frac{{\beta^*}^{\intercal} \Sigma^W \beta^*}{m}} \cdot Z_{\ell} + o_p(1/\sqrt{m})
\end{equation*}
where $Z_{\ell}$ for $\ell = 1, \dots, L$ are independent standard Gaussian variables. Then, 
\begin{align*}
    \left(\hat{\mathbb{E}}^0[\phi_{\ell}] - \sum_{k=1}^{K} \hat{\beta}_k \hat{\mathbb{E}}^k[\phi_{\ell}] \right)^2 & = \left(\left(\hat{\mathbb{E}}^0[\phi_{\ell}] - \sum_{k=1}^{K} \beta^*_k\hat{\mathbb{E}}^k[\phi_{\ell}]\right) + \left( \sum_{k=1}^{K}(\hat{\beta}_k - \beta^*_k)(\hat{\mathbb{E}}^0[\phi_{\ell}] - \hat{\mathbb{E}}^k[\phi_{\ell}])\right) \right)^2 \\
    & \stackrel{d}{=}{\frac{{\beta^*}^{\intercal} \Sigma^W \beta^*}{m}} \cdot Z_{\ell}^2 + o_p(1/m).
\end{align*}
Therefore, we may estimate ${\beta^*}^{\intercal} \Sigma^W \beta^* / m$ as 
\begin{equation}\label{eq:ci-part2}
    \frac{1}{L}\sum_{\ell = 1}^{L}  \left(\hat{\mathbb{E}}^0[\phi_{\ell}] - \sum_{k=1}^{K} \hat{\beta}_k \hat{\mathbb{E}}^k[\phi_{\ell}] \right)^2 \stackrel{d}{=} {\frac{{\beta^*}^{\intercal} \Sigma^W \beta^*}{m}} \cdot \frac{\chi^2(L)}{L} + o_p(1/m),
\end{equation}
where $\chi^2(L)$ is a chi-squared random variable with degrees of freedom $L$. 

Finally, as $L \xrightarrow[]{}\infty$, we have the consistency of $\hat{\beta}$ and $\chi^2(L)/L \xrightarrow[]{p} 1$. Suppose we have $\widehat{\text{Var}}_{\mathbb{P}^0}(\phi_0(D))  = \text{Var}_{\mathbb{P}^0}(\phi_0(D)) + o_p(1)$ as in Remark~\ref{remark:var}. Combining all the results, we have asymptotically valid ($1-\alpha$) confidence interval as $m \xrightarrow[]{} \infty$ and $L \xrightarrow[]{} \infty$:
\begin{equation*}
    \sum_{k=1}^{K} \hat{\beta}_k \hat{\mathbb{E}}^k[\phi_0] \pm z_{1-\alpha/2} \cdot  \sqrt{\widehat{\text{Var}}_{\mathbb{P}^0}(\phi_0)} \cdot \sqrt{ \frac{1}{L}\sum_{\ell = 1}^{L}  \left(\hat{\mathbb{E}}^0[\phi_{\ell}] - \sum_{k=1}^{K} \hat{\beta}_k \hat{\mathbb{E}}^k[\phi_{\ell}] \right)^2 }.
\end{equation*}

Now let's consider the case in Remark~\ref{remark:ci-target} where 
\begin{equation*}
    \theta^0 = \arg\min_\theta \mathbb{E}^0[\mathcal{L}(\theta, X, Y)], \quad \hat{\theta} = \arg\min_\theta \sum_{k=1}^K \hat{\beta}_k \hat{\mathbb{E}}^k[\mathcal{L}(\theta, X, Y)].
\end{equation*}
Assume the regularity conditions of Lemma~\ref{lemma:ooderror} and the consistency of $\hat{\beta}$ and $\hat{\theta}$ ($\hat{\beta} \xrightarrow[]{p} \beta^*$,  $\hat{\theta} \xrightarrow[]{p} \theta^0$). Then, from the proof of Lemma~\ref{lemma:ooderror}, we have
\begin{equation*}
    \sqrt{m}\left(\hat{\theta} - \theta^0\right) = \sqrt{m}\left(\sum_{k=1}^K \hat{\beta}_k \hat{\mathbb{E}}^k[\phi(D)] - \mathbb{E}^0[\phi(D)]\right) + o_P(1),
\end{equation*}
where $\phi(D) = - \mathbb{E}^0[\partial_\theta^2 \mathcal{L}(\theta^0,D)]^{-1}  \partial_\theta \mathcal{L}(\theta^0,D)$. Then, from the results above, we only need to show that $\widehat{\text{Var}}_{\mathbb{P}^0}(\phi(D))  = \text{Var}_{\mathbb{P}^0}(\phi(D)) + o_p(1)$. Note that
\begin{align*}
    \sum_{k=1}^K \hat \beta_k \hat{\mathbb{E}}^k[\partial_\theta^2 \mathcal{L}( \hat \theta,X,Y)] &- \mathbb{E}^0[\partial_\theta^2 \mathcal{L}( \theta^0,X,Y)] \\&=  \sum_{k=1}^K \hat \beta_k \hat{\mathbb{E}}^k[\partial_\theta^2 \mathcal{L}( \hat \theta,X,Y)] - \sum_{k=1}^K \hat \beta_k \hat{\mathbb{E}}^k[\partial_\theta^2 \mathcal{L}( \theta^0,X,Y)] \\
    & + \sum_{k=1}^K \hat \beta_k \hat{\mathbb{E}}^k[\partial_\theta^2 \mathcal{L}( \theta^0,X,Y)] - \sum_{k=1}^K \beta_k^* {\mathbb{E}}^0[\partial_\theta^2 \mathcal{L}( \theta^0,X,Y)]. 
\end{align*}
Similarly as in the proof of Lemma~\ref{lemma:ooderror}, we have that
\begin{align*}
    \Big | \Big|\sum_{k=1}^K \hat \beta_k \hat{\mathbb{E}}^k[\partial_\theta^2 \mathcal{L}( \hat \theta,X,Y)] &- \sum_{k=1}^K \hat \beta_k \hat{\mathbb{E}}^k[\partial_\theta^2 \mathcal{L}( \theta^0,X,Y)] \Big | \Big| \\ & \leq \left(\sum_{k=1}^K (|\beta_k^*|+o_P(1)) \frac{1}{n_k}\sum_{i=1}^{n_k}h(D_{ki})\right) ||\hat{\theta}-\theta^0||= o_P(1).
\end{align*}
Moreover, with the consistency of $\hat{\beta}$ and Theorem~\ref{thm:perturbation-model}, we have 
\begin{equation*}
    \sum_{k=1}^K \hat \beta_k \hat{\mathbb{E}}^k[\partial_\theta^2 \mathcal{L}( \theta^0,X,Y)] - \sum_{k=1}^K \beta_k^* {\mathbb{E}}^0[\partial_\theta^2 \mathcal{L}( \theta^0,X,Y)] = o_P(1)
\end{equation*}
Therefore, 
\begin{equation*}
    \hat{\phi}(D) = -\left(\mathbb{E}^0[\partial_\theta^2 \mathcal{L}( \theta^0,D)] + o_P(1)\right)^{-1}\partial_\theta \mathcal{L}(\hat{\theta}, D). 
\end{equation*}
Note that 
\begin{align*}
    \hat{\mathbb{E}}^k[\partial_\theta \mathcal{L}(\hat{\theta}, D)\partial_\theta \mathcal{L}(\hat{\theta}, D)^{\intercal}] &= \hat{\mathbb{E}}^k[(\partial_\theta \mathcal{L}({\theta}^0, D) + \delta(\hat{\theta}, D))(\partial_\theta \mathcal{L}(\theta^0, D)  + \delta(\hat{\theta}, D))^{\intercal}] \\
    & \leq   \hat{\mathbb{E}}^k[\partial_\theta \mathcal{L}({\theta}^0, D)\partial_\theta \mathcal{L}({\theta}^0, D)^{\intercal}]  +  \hat{\mathbb{E}}^k[\delta(\hat{\theta}, D)\delta(\hat{\theta}, D)^{\intercal}] \\
    & + 2 \sqrt{\hat{\mathbb{E}}^k[\partial_\theta \mathcal{L}({\theta}^0, D)\partial_\theta \mathcal{L}({\theta}^0, D)^{\intercal}] \cdot \hat{\mathbb{E}}^k[\delta(\hat{\theta}, D)\delta(\hat{\theta}, D)^{\intercal}] }
\end{align*}
where $\delta(\hat{\theta}, D) = \partial_\theta \mathcal{L}(\hat{\theta}, D)- \partial_\theta \mathcal{L}({\theta}^0, D)$. The inequality is from the Cauchy-Schwarz inequality. Similarly as in the proof of Lemma~\ref{lemma:ooderror} with Taylor expansion, we have 
\begin{align*}
    \hat{\mathbb{E}}^k[\delta(\hat{\theta}, D)\delta(\hat{\theta}, D)^{\intercal}] = o_P(1).
\end{align*}
and  we have $\hat{\mathbb{E}}^k[\partial_\theta \mathcal{L}({\theta}^0, D)\partial_\theta \mathcal{L}({\theta}^0, D)^{\intercal}] = {\mathbb{E}}^0[\partial_\theta \mathcal{L}({\theta}^0, D)\partial_\theta \mathcal{L}({\theta}^0, D)^{\intercal}] + o_P(1)$ from Thereom~\ref{thm:perturbation-model}.
Then, we get 
\begin{align*}
     \hat{\mathbb{E}}^k[\partial_\theta \mathcal{L}(\hat{\theta}, D)\partial_\theta \mathcal{L}(\hat{\theta}, D)^{\intercal}]  \leq   \hat{\mathbb{E}}^k[\partial_\theta \mathcal{L}({\theta}^0, D)\partial_\theta \mathcal{L}({\theta}^0, D)^{\intercal}]  + o_P(1).
\end{align*}
By using Cauchy-Schwarz inequality in other direction as well, we get
\begin{align*}
     \hat{\mathbb{E}}^k[\partial_\theta \mathcal{L}(\hat{\theta}, D)\partial_\theta \mathcal{L}(\hat{\theta}, D)^{\intercal}]  =   \hat{\mathbb{E}}^k[\partial_\theta \mathcal{L}({\theta}^0, D)\partial_\theta \mathcal{L}({\theta}^0, D)^{\intercal}]  + o_P(1).
\end{align*}
Similarly, we can derive
\begin{align*}
    \hat{\mathbb{E}}^k[\partial_\theta \mathcal{L}(\hat{\theta}, D)] &=  \hat{\mathbb{E}}^k[\partial_\theta \mathcal{L}({\theta}^0, D)]+ o_P(1).
\end{align*}
Then, the variance estimate on the pooled donor data can be written as
\begin{align*}
    \widehat{\text{Var}}_{\mathbb{P}^0}(\phi) &= \sum_{k=1}^K \frac{n_k}{n}\hat{\mathbb{E}}^k[\hat{\phi}(D) \hat{\phi(D)}^{\intercal}] - \left( \sum_{k=1}^K \frac{n_k}{n}\hat{\mathbb{E}}^k[\hat{\phi}(D)]\right) \left( \sum_{k=1}^K \frac{n_k}{n}\hat{\mathbb{E}}^k[\hat{\phi}(D)]\right)^{\intercal}\\&= \sum_{k=1}^K \frac{n_k}{n}\hat{\mathbb{E}}^k[{\phi}(D) {\phi(D)}^{\intercal}] - \left( \sum_{k=1}^K \frac{n_k}{n}\hat{\mathbb{E}}^k[{\phi}(D)]\right) \left( \sum_{k=1}^K \frac{n_k}{n}\hat{\mathbb{E}}^k[{\phi}(D)]\right)^{\intercal} + o_P(1).
\end{align*}
Then, by the proof in the Appendix, \ref{sec:remark-2}, we have $\widehat{\text{Var}}_{\mathbb{P}^0}(\phi(D))  = \text{Var}_{\mathbb{P}^0}(\phi(D)) + o_p(1)$. Finally, as $L \xrightarrow[]{} \infty$, we have the consistency of $\hat{\beta}$ and under the regularity conditions of Lemma~\ref{lemma:consistency}, we have the consistency of $\hat{\theta}$ by following the proof of Lemma~\ref{lemma:consistency}. This completes the proof.

\section{Additional Details in Data Analysis}\label{sec:add-details}

\subsection{Covariance of Test Functions in GTEx Data Analysis}

We estimate the sparse inverse covariance matrix of test functions $(\phi_{\ell})_{\ell = 1, \dots, L = 1000}$ defined in Section \ref{sec:gtex}, using group Lasso with a Python package \texttt{sklearn.covariance.GraphicalLassoCV}. The heatmap of fitted precision matrix is given in Figure~\ref{fig:sparse-covariance}. As the estimated inverse covariance matrix was close to a diagonal matrix, we did not perform a whitening transformation for the 1000 test functions. 

\begin{figure}[ht]
    \centering
    \includegraphics[scale = 0.6]{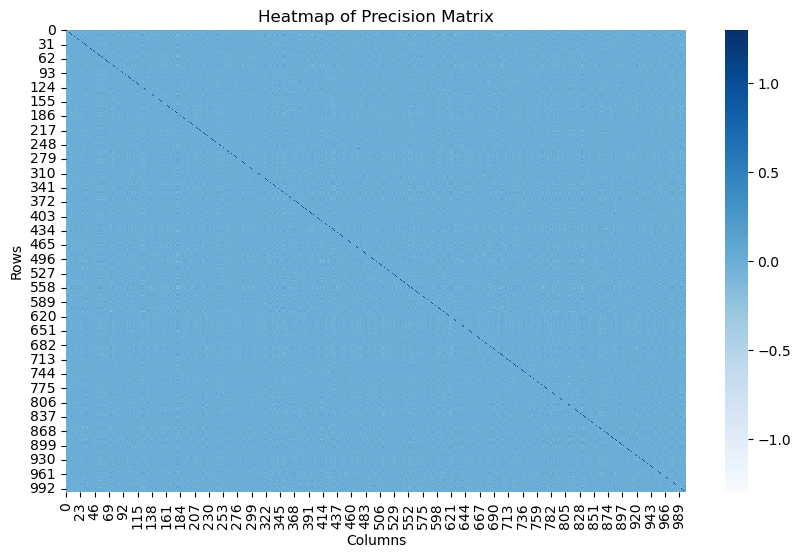}
    \caption{Fitted sparse precision matrix of test functions defined in Section \ref{sec:gtex} on the GTEx data}
    \label{fig:sparse-covariance}
\end{figure}

\subsection{Additional Results for ACS Income Data Analysis}

In this section, we provide additional results for ACS income data analysis in Section~\ref{sec:acs} for 6 randomly selected target states. The results are presented in Figure~\ref{fig:ACS-appendix}.

\begin{figure}
    \centering
    \includegraphics[scale = 0.5]{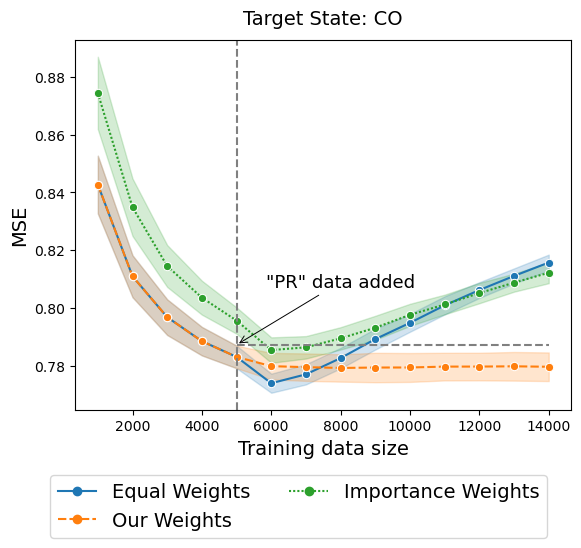}
    \includegraphics[scale = 0.5]{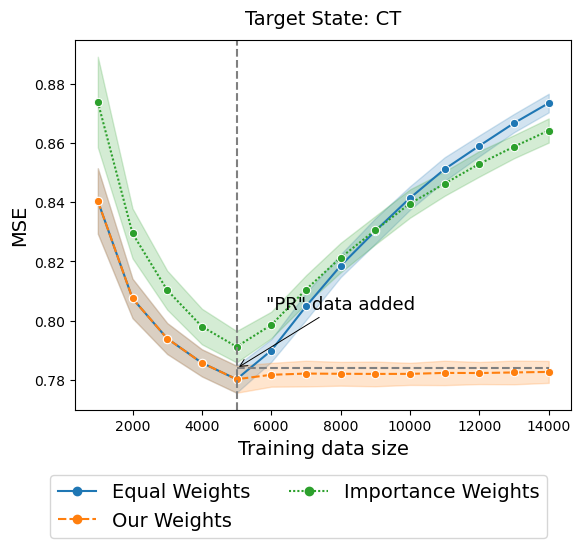}
    \includegraphics[scale = 0.5]{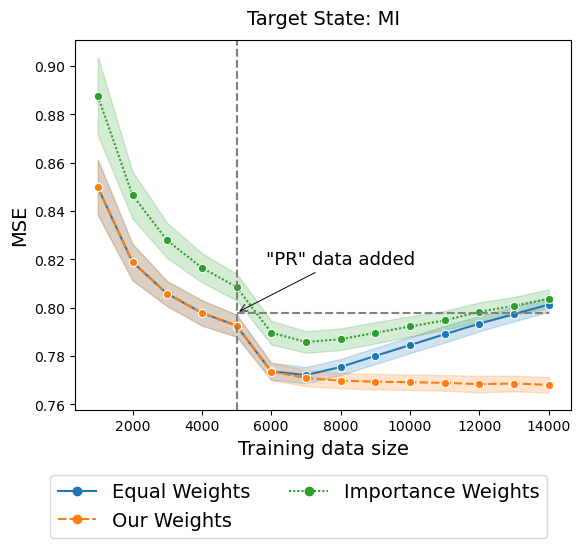}
    \includegraphics[scale = 0.5]{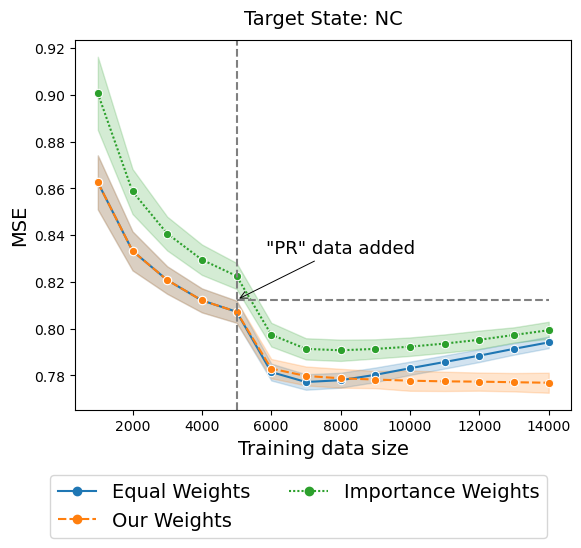}
    \includegraphics[scale = 0.5]{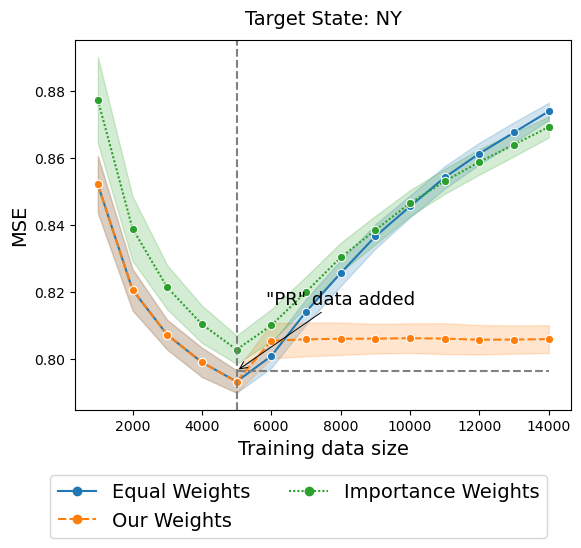}
    \includegraphics[scale = 0.5]{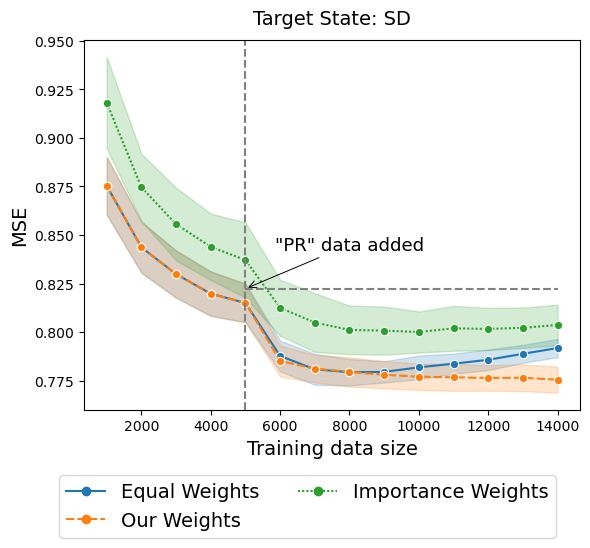}
    \caption{Test MSE results of the XGBoost when the training data is initially sourced from CA and then sourced from PR. The dashed vertical line indicates the point when PR data started to be added. The blue line is when samples are equally weighted, and the green line is when samples are weighted based on importance weights. The orange line is when samples receive distribution-specific weights using our proposed method.}
    \label{fig:ACS-appendix}
\end{figure}

\end{document}